\documentclass[journal]{IEEEtran}
\usepackage{cite}
\usepackage{amsmath,amssymb,amsfonts}

\usepackage{graphicx}
\usepackage{textcomp}
\usepackage{xcolor}
\def\BibTeX{{\rm B\kern-.05em{\sc i\kern-.025em b}\kern-.08em
    T\kern-.1667em\lower.7ex\hbox{E}\kern-.125emX}}
    
\usepackage[dvipsnames]{xcolor}
\definecolor{colorblindfree3_1}{RGB}{252,141,89}
\definecolor{colorblindfree3_2}{RGB}{255,255,191}
\definecolor{colorblindfree3_3}{RGB}{145,191,219}

\definecolor{colorblindfree4_1}{RGB}{215,25,28}
\definecolor{colorblindfree4_2}{RGB}{253,174,97}
\definecolor{colorblindfree4_3}{RGB}{171,217,233}
\definecolor{colorblindfree4_4}{RGB}{44,123,182}

\definecolor{colorblindfree5_1}{RGB}{215,25,28}
\definecolor{colorblindfree5_2}{RGB}{253,174,97}
\definecolor{colorblindfree5_3}{RGB}{255,255,191}
\definecolor{colorblindfree5_4}{RGB}{171,217,233}
\definecolor{colorblindfree5_5}{RGB}{44,123,182}

\definecolor{colorblindfree6_1}{RGB}{215,48,39}
\definecolor{colorblindfree6_2}{RGB}{252,141,89}
\definecolor{colorblindfree6_3}{RGB}{254,224,144}
\definecolor{colorblindfree6_4}{RGB}{224,243,248}
\definecolor{colorblindfree6_5}{RGB}{145,191,219}
\definecolor{colorblindfree6_6}{RGB}{69,117,180}

\definecolor{colorblindfree7_1}{RGB}{215,48,39}
\definecolor{colorblindfree7_2}{RGB}{244,109,67}
\definecolor{colorblindfree7_3}{RGB}{253,174,97}
\definecolor{colorblindfree7_4}{RGB}{254,224,144}
\definecolor{colorblindfree7_5}{RGB}{224,243,248}
\definecolor{colorblindfree7_6}{RGB}{171,217,233}
\definecolor{colorblindfree7_7}{RGB}{116,173,209}

\definecolor{colorblindfree8_1}{RGB}{215,48,39}
\definecolor{colorblindfree8_2}{RGB}{244,109,67}
\definecolor{colorblindfree8_3}{RGB}{253,174,97}
\definecolor{colorblindfree8_4}{RGB}{254,224,144}
\definecolor{colorblindfree8_5}{RGB}{224,243,248}
\definecolor{colorblindfree8_6}{RGB}{171,217,233}
\definecolor{colorblindfree8_7}{RGB}{116,173,209}
\definecolor{colorblindfree8_8}{RGB}{69,117,180}

\definecolor{colorblindfree9_1}{RGB}{215,48,39}
\definecolor{colorblindfree9_2}{RGB}{244,109,67}
\definecolor{colorblindfree9_3}{RGB}{253,174,97}
\definecolor{colorblindfree9_4}{RGB}{254,224,144}
\definecolor{colorblindfree9_5}{RGB}{255,255,191}
\definecolor{colorblindfree9_6}{RGB}{224,243,248}
\definecolor{colorblindfree9_7}{RGB}{171,217,233}
\definecolor{colorblindfree9_8}{RGB}{116,173,209}
\definecolor{colorblindfree9_9}{RGB}{69,117,180}

\definecolor{colorblindfree10_1}{RGB}{165,0,38}
\definecolor{colorblindfree10_2}{RGB}{215,48,39}
\definecolor{colorblindfree10_3}{RGB}{244,109,67}
\definecolor{colorblindfree10_4}{RGB}{253,174,97}
\definecolor{colorblindfree10_5}{RGB}{254,224,144}
\definecolor{colorblindfree10_6}{RGB}{224,243,248}
\definecolor{colorblindfree10_7}{RGB}{171,217,233}
\definecolor{colorblindfree10_8}{RGB}{116,173,209}
\definecolor{colorblindfree10_9}{RGB}{69,117,180}
\definecolor{colorblindfree10_10}{RGB}{49,54,149}

\definecolor{colorblindfree11_1}{RGB}{165,0,38}
\definecolor{colorblindfree11_2}{RGB}{215,48,39}
\definecolor{colorblindfree11_3}{RGB}{244,109,67}
\definecolor{colorblindfree11_4}{RGB}{253,174,97}
\definecolor{colorblindfree11_5}{RGB}{254,224,144}
\definecolor{colorblindfree11_6}{RGB}{255,255,191}
\definecolor{colorblindfree11_7}{RGB}{224,243,248}
\definecolor{colorblindfree11_8}{RGB}{171,217,233}
\definecolor{colorblindfree11_9}{RGB}{116,173,209}
\definecolor{colorblindfree11_10}{RGB}{69,117,180}
\definecolor{colorblindfree11_11}{RGB}{49,54,149}

\usepackage[caption=false,font=normalsize,labelfont=sf,textfont=sf]{subfig}
\usepackage{pgfplots}
\pgfplotsset{compat=1.18, every axis plot/.append style={line width=0.8pt}}
\usepackage{bm}
\usepackage{enumitem}
\usepackage{bbm}%for \mathbb{1}
\usepackage[ruled, linesnumbered, vlined]{algorithm2e}
\usepackage{booktabs} %rule in table
\usetikzlibrary{spy,backgrounds}
\usepackage{placeins}%Control reference in the last page, float barrier
\usepackage{amssymb}
\usepackage{mathtools}
\usepackage{amsthm}
\usepackage[inkscapeformat=png]{svg}
\usepackage{nccmath}
\usepackage[acronym]{glossaries}
\usepackage{comment}
\usepackage{amsfonts}
\usepackage{bbm}
\usepackage[normalem]{ulem}

\DeclarePairedDelimiter\floor{\lfloor}{\rfloor}

\newacronym{sc}{SC}{successive cancellation}
\newacronym{scl}{SCL}{successive cancellation list}
\newacronym{ca}{CA}{cyclic-redundancy-check aided}
\newacronym{aed-sc}{AED-SC}{automorphism ensemble decoding with the successive cancellation constituent decoder}
\newacronym{crc}{CRC}{cyclic-redundancy-check}
\newacronym{pe}{PE}{perturbation-enhanced}
\newacronym{be}{BE}{bias-enhanced}
\newacronym{pm}{PM}{path metric}
\newacronym{rpe}{RPE}{random perturbation-enhanced}
\newacronym{fer}{FER}{frame error rate}
\newacronym{awgn}{AWGN}{additive white Gaussian noise}
\newacronym{bpsk}{BPSK}{binary phase shift keying}
\newacronym{qpsk}{QPSK}{quadratic phase shift keying}
\newacronym{bec}{BEC}{binary erasure channel}
\newacronym{bms}{BMS}{binary-input memoryless}
\newacronym{aga}{AGA}{approximated Gaussian approximation}
\newacronym{snr}{SNR}{signal-to-noise ratio}
\newacronym{llr}{LLR}{log-likelihood ratio}
\newacronym{gpscl}{GPSCL}{generalized partitioned successive cancellation list}
\newacronym{spc}{SPC}{single parity-check}
\newacronym{ida}{IDA}{input-distribution-aware}
\newacronym{ldpc}{LDPC}{low-density parity-check}
\newacronym{ml}{ML}{maximum likelihood}

\newtheorem{theorem}{Theorem}

\newtheorem{prop}{Proposition}

\newenvironment{customlegend}[1][]{%
    \begingroup
    \csname pgfplots@init@cleared@structures\endcsname
    \pgfplotsset{#1}%
}{%
    \csname pgfplots@createlegend\endcsname
    \endgroup
}%
\def\addlegendimage{\csname pgfplots@addlegendimage\endcsname}
    
\begin{document}

\title{Enhanced Successive Cancellation List Decoder for Long Polar Codes Targeting Air Interface\\
{
%\footnotesize \textsuperscript{*}Note: Sub-titles are not captured for https://ieeexplore.ieee.org  and should not be used}
%\thanks{Identify applicable funding agency here. If none, delete this.
}
}

\author{
\IEEEauthorblockN{Jiajie Li, Sihui Shen and Warren J. Gross\\}
%IEEE Publication Technology,~\IEEEmembership{Staff,~IEEE,}
        % <-this % stops a space
%\IEEEauthorblockA{Department of Electrical and Computer Engineering,
%McGill University\\  Montr\'eal, Qu\'ebec, Canada\\
%Emails: \{jiajie.li, sihui.shen\}@mail.mcgill.ca, warren.gross@mcgill.ca}
%\thanks{This paper was produced by the IEEE Publication Technology Group. They are in Piscataway, NJ.}% <-this % stops a space
%\thanks{Manuscript received April 19, 2021; revised August 16, 2021.}
\thanks{Jiajie Li, Sihui Shen, and Warren J. Gross are with the Department of Electrical and Computer Engineering, McGill University, Montr{\'e}al, Qu{\'e}bec, H3A 0E9, Canada (e-mail: jiajie.li@mail.mcgill.ca; sihui.shen@mail.mcgill.ca;warren.gross@mcgill.ca).}
}

\maketitle

\begin{abstract}
Polar codes are the first codes with a proven capacity-achieving capability, but their decoding faces several challenges, especially under long code lengths.
In this paper, we target algorithmic improvements and analyses to enable the implementation of long polar codes (e.g., length 8K bits) by addressing key challenges in memory usage and computational complexity presented by successive cancellation list (SCL) polar decoding.
Perturbation-enhanced (PE) SCL decoders with a list size of $L$ reach the decoding performance of the SCL decoder with a list size of $2L$.
The proposed bias-enhanced (BE) SCL decoders,  which simplify the PE SCL decoder based on insights gained by an ablation study, return similar decoding performance to PE SCL decoders.
Also,  proposed BE generalized partitioned SCL (GPSCL) decoders with a list size of $8$ have a $67\%$ reduction in the memory usage and similar decoding performance compared to SCL decoders with a list size of $16$, and it demonstrates that an accurate bias can be generated under a reduced number of codewords from the list and reduces the overhead from $\left(L-1\right)n$ XOR gates plus $n$ priority encoders to $n$ XOR gates, where $n$ is the code length.
Furthermore, input-distribution-aware (IDA) decoding is applied to BE GPSCL decoders, which shows how an accurate bias is generated under a low-complexity decoder. 
Up to $5.4\times$ reduction in the computational complexity is achieved compared to SCL decoders with a list size of $16$, and negligible latency overhead is added to the decoding process.
The degraded decoding performance is at most $0.05\text{ dB}$ compared to BE GPSCL decoders without IDA decoding.
Lastly, we theoretically prove that the bias in the BE SCL decoder moves the received soft information toward valid polar codewords with a high likelihood, and explain the decoding performance gain. 
\end{abstract}

\begin{IEEEkeywords}
ablation study, polar codes, perturbation-enhanced decoding, successive cancellation list decoder
\end{IEEEkeywords}

\section{Introduction}
Emerging applications are introducing new challenges for future communication systems, particularly in terms of reducing energy consumption and implementation costs~\cite{andrew20246g,sambhwani20226g,fettweis20216g,zhang20236g,geiselhart20236g}. 
These constraints drive the need for efficient and scalable deployment strategies, especially in scenarios involving limited infrastructure~\cite{chen2020connected}.
To meet these demands, reduced-capability devices are required to support a broad range of use cases. 
The air interface remains a critical component of communication systems, with improvements proposed in~\cite{vetter20236g,viswanathan20206g}. 
Channel coding, a key part of the air interface~\cite{vetter20236g}, plays a significant role in ensuring reliable data transmission. 
Research is needed to develop decoding techniques that maintain a high reliability while operating under a reduced computational capability, with a focus on minimizing energy usage and implementation costs to address the requirements of emerging applications.

Polar codes are the first channel coding scheme with the proven capacity-achieving capability~\cite{Polar}.
Though polar codes can be decoded by the low complexity ($O\left(n\log\left(n\right)\right)$~\cite{Polar}) \gls{sc} decoder, the decoding performance is mediocre under short-to-medium code lengths as the capacity-achieving capability can only be achieved under the asymptotical case.
Polar codes (along with a cyclic redundancy check outer code) have been adopted in the 5G communication standard~\cite{embb,Bioglio20215gpolarttorial} and exhibit excellent decoding performance in the short-to-medium code length (up to 2K bits) using the \gls{scl} decoder~\cite{SCL}, and polar codes with a long code length is suggested to used in new applications~\cite{geiselhart20236g}. 

Polar codes exhibit improving performance as the code length increases and achieve capacity asymptotically in the code length~\cite{Polar}, but the \gls{scl} decoder~\cite{SCL} is still needed to achieve a good decoding performance under a code length of up to 8K for applications that require a high reliability.
Hence, there remains a need for research to address important challenges in the implementation of decoders for long polar codes, especially targeting the stringent requirements on energy consumption and the chip area that is related to the implementation cost.
In this work, we focus on reducing computational complexity (number of operations) and memory usage, metrics which are closely related to energy consumption and chip area~\cite{zhang20236g}~\cite{Alexios2015llrscl}.

For long polar codes, the area requirement of implementing the \gls{sc} decoder is the main challenge, and has been addressed in prior work by dividing the decoder into two phases~\cite{Pamuk2013twophasesSC}, using a semi-parallel architecture~\cite{Leroux2013semiparalleSC}, mixing different decoder architectures~\cite{Le2016threephasesSC}, rate-dependent quantization and chaining optimization~\cite{Raymond2014ScableSC}, or non-uniform quantization of internal log-likelihood ratio and compressing frozen bit memory~\cite{Le2017SCverylongpolar}.
The fast \gls{sc} decoder is implemented for decoding long polar codes~\cite{FSC}, and it uses fewer look-up tables and registers but more random access memory than the semi-parallel SC~\cite{Leroux2013semiparalleSC} in the field programmable gate array implementation.

Given the large area requirement of the SC decoder for long polar codes, the area requirement of the \gls{scl} for long polar codes will be even larger.
For example, extra memory, which needs a large area, is required to store information for the list decoding, and the required sorter increases the area needed in the hardware implementation~\cite{MetricSortSCL}.
For \gls{scl} decoders, most implementations are focused on code length $\leq 2^{11}$~\cite{liu2018length2048scl}.
An \gls{scl} decoder for length-$2^{12}$ polar codes is built for a storage system that uses hard decision values in~\cite{park2024longpolarstorage}.
An improved \gls{gpscl} decoder with different list sizes in different partitions is proposed for polar codes with a code length of $2^{13}$ in~\cite{Shen2020listdecoderlongpolar}.
Prior work~\cite{Shen2020listdecoderlongpolar} focuses on algorithmic improvement on latency reduction.
In this work, we are aiming at reducing the memory usage and the computational complexity, and improving the decoding performance of the SCL decoder for decoding polar codes with a length $>2^{11}$.

Recently, a \gls{pe} \gls{scl} decoder was proposed to improve the error correction performance of the \gls{scl} decoder~\cite{wang2023randomperlist,wang2024adaptive}, and it was proved that, asymptotically, the decoding performance can be improved~\cite{Liu2025analysisperturbationsc}.
When decoding long polar codes, the \gls{pe} \gls{scl} decoder with a list size of $L$ can achieve a similar error correction performance to the \gls{scl} decoder with a list size of $2L$, either using ten decoding attempts with random perturbation or using two decoding attempts with adaptive and fixed amount perturbation based on the previous \gls{scl} decoding result, which makes the adaptive \gls{pe} decoder a good candidate for practical implementation.
The perturbation decoding has also been extended to the \gls{sc} decoder~\cite{yang2025biassuccessive,yang2025improved,yang2026improved}, the \gls{sc} flip decoder~\cite{pillet2025successiveflip}, and the \gls{scl} flip decoder~\cite{yang2026PSCLF}, but their decoding performance does not reach the \gls{scl} decoder with a list size of $16$, unlike the \gls{pe} \gls{scl} decoder.
In this work, we investigate the key factor to the improvement in the \gls{pe} \gls{scl} decoder and design simplified and low-complexity enhancement-based \gls{scl} decoders.

Portions of this work have been accepted by the 2025 Asilomar Conference on Signals, Systems, and Computers, where the ablation study (i.e., removing parts of the algorithm to investigate how removed parts affect the decoding performance) on the PE SCL decoder is performed, and a simplified algorithmic design, the \gls{be} SCL decoder, is proposed~\cite{li2025bescl}.
The proposed \gls{be} \gls{scl} decoder can return similar error correction performance to the \gls{pe} \gls{scl} decoder when the number of decoding attempts is $2$.
A recent work~\cite{yang2025biassuccessive} analyzes the bias added to the \gls{sc} decoding, and it shows that added bias moves received symbols closer to the maximum likelihood estimations.
Building on top of the prior work~\cite{li2025bescl}, we target algorithmic improvements, empirical and theoretical analyses, and make the following contributions:

\begin{enumerate}
    \item We use a partitioned decoder to reduce the memory usage of the \gls{be} \gls{scl} decoder, which helps to enable the implementation of the \gls{scl} decoder for long polar codes.
    Compared to the \gls{scl} decoder with a list size of $16$, a $67\%$ reduction in the memory usage is achieved without degraded decoding performance.
    \item We use \gls{ida} decoding to adaptively determine the list size for the \gls{be} decoder, which reduces the average list size used by the decoding.
    When compared to the \gls{scl} decoder with a list size of $16$, the reduced average list size leads to a reduction of up to $5.4\times$ in the average computational complexity.
    The degradation in performance is at most $0.05\text{ dB}$ compared to the \gls{be} decoders without \gls{ida} decoding.
    \item The approximated \gls{pm} used by the fast \gls{scl} decoder~\cite{hashemi2017FastSCL} is analyzed in this work. We prove that the error due to the approximation can be calibrated without the computational-complexity increase.
    We prove that the bias in the \gls{be} \gls{scl} decoder moves the received soft information toward valid polar codewords with a high likelihood, and explain the decoding performance gain.
\end{enumerate}
While existing techniques (i.e., \gls{gpscl} decoders~\cite{Hashemi2018partitionlist} and the \gls{ida} decoding~\cite{IDASCL}) are applied in this work, the following new contributions are also made.
First, by using the \gls{gpscl} decoder, we found that an accurate bias can be generated under a reduced number of codewords from the list, which can reduce the overhead brought by the enhancement-based decoder from $\left(L-1\right)n$ XOR gates plus $n$ $L$-to-$\log_{2}\left(L\right)$ priority encoders to $n$ XOR gates, where $n$ is the code length.
Secondly, by using the \gls{ida} decoding, we can see that a low-complexity decoder can generate an accurate bias, and negligible latency overhead is required by this added \gls{ida} decoding.

Furthermore, the adaptive decoding idea has been applied in literature like~\cite{li2012adaptive,song2019efficient}.
Compared to prior works~\cite{li2012adaptive,song2019efficient}, our work has new contributions in the following way.
In~\cite{li2012adaptive,song2019efficient}, the decoding starts with a small list/stack size and then gradually increases the list/stack size if the decoding using the smaller list/stack size fails or certain criteria are met.
Though a low computational complexity is returned, the best possible decoding performance is capped by the maximum list/stack size allowed, and the implementation should have enough resources to support the decoding with the maximum list/stack size.
However, the \gls{be} decoding can improve the decoding performance under a reduced computational complexity, which is similar to~\cite{li2012adaptive,song2019efficient}, but only requires a smaller amount of resources in the implementation to achieve the target \gls{fer}, unlike~\cite{li2012adaptive,song2019efficient}.
Also, the adaptive mechanism in~\cite{song2019efficient} is integrated into the decoding process, which causes extra latency overhead, while our proposed systematic integration causes negligible latency overhead.

This work is structured as follows. 
Section~\ref{sec:preliminaries} provides the necessary background of polar codes and decoders used in this work. 
Section~\ref{sec:ablation} shows the ablation study of the adaptive \gls{pe} \gls{scl} decoder. 
Experimental results of the \gls{be} \gls{scl} decoder are shown in Section~\ref{sec:results}.
The BE decoder with the reduced memory usage is presented in Section~\ref{sec:redcuedmem}.
The BE decoder with the reduced computational complexity is presented in Section~\ref{sec:reducedcomp}.
%and the analysis in the path metric are presented in Section~\ref{sec:reducedcomp}.
Section~\ref{sec:ana_pm_bias} analyzes the \gls{pm} used by the \gls{scl} decoding and the bias used in the \gls{be} \gls{scl} decoder.
The conclusion is drawn in Section~\ref{sec:conclusion}.

\section{Preliminaries}
Matrices and vectors are denoted by bold upper-case letters ($\bm{M}$) and lower-case letters ($\bm{m}$).
The $i$th element of the vector $\bm{m}$ is $m_{i}$.
The Kronecker product, the matrix transpose, the sign function, the absolute value, and the hard decision $\mathbbm{1}\left(x<0\right)$ are denoted by $\otimes$, $^\top$, $\text{sign}(\cdot)$, $|\cdot|$, and $\text{HD}\left(\cdot\right)$.
\label{sec:preliminaries}
\subsection{Constructions of Polar Codes}
The generator matrix of length-$n$ polar codes is constructed by taking the $m$-th Kronecker power of a base matrix:
\begin{equation}
    \bm{G}=\bm{F}^{\otimes m}, \text{ }\bm{F}=
    \begin{bmatrix}
    1&0\\
    1&1
    \end{bmatrix}\text{,}
\end{equation}
where $m = \log_{2}(n)$, $n$ is the code length, and $\bm{F}$ is the base matrix.
To encode $k$ information bits, polar codes select $k$ rows in the generator matrix $\bm{G}$ to place information bits, selected rows correspond to synthetic channels with a high reliability, and we call these synthetic channels the information set $\mathcal{I}$.
The remaining $n-k$ rows are in the frozen set $\mathcal{F}$.
The code rate $R$ is defined as $R=k/n$.
For \gls{ca} polar codes, a length-$n_{\text{crc}}$ \gls{crc} code is appended to the message, so the number of encoded message bits is reduced to $k-n_{\text{crc}}$, and the effective code rate is $R=(k-n_{\text{crc}})/n$.
In this work, the reliability order of the synthetic channel is computed by the $\beta$-expansion method~\cite{BetaExpansionPolar}, 
%\sout{which is adopted by the 5G communication standard~\cite{embb},} 
with a length of up to $2^{13}=8192$, which is smaller than the largest code length ($n=8448$~\cite{embb}) supported by the control and the data channel of the 5G communication standard, and different effective code rates, and a \gls{crc} outer code is used.

\subsection{Successive Cancellation Decoder}

This work uses the convention that estimations of message bits are in the first stage of the factor graph for polar codes, and the estimation of the transmitted codeword is in the $(m+1)$-th stage of the factor graph for polar codes.
Message bits $1$ to $n$ are determined sequentially in the \gls{sc} decoder, and the message bit $u_{i}$ is estimated according to all received soft information and previously estimated message bits $u_{1},u_{2},...,u_{i-1}$.
The estimation of the message bits in $\mathcal{F}$ is always $0$.
The decoding process can be explicitly viewed as first recursively updating soft information from stage-$\left(m + 1\right)$ to stage-$1$~\cite{PolarDecodersReview}
\begin{equation}
%\resizebox{\columnwidth}{!}{
\begin{split}
    &l_{i,j} = \\
    &\begin{cases}
        2\tanh^{-1}{\left(\tanh{\left(\frac{l_{i+1,j}}{2}\right)}\tanh{\left(\frac{l_{i+1,j + \xi}}{2}\right)}\right)}\text{,} \floor{\frac{j - 1}{\xi}}\text{mod}2 = 0\text{,}\\
        (1 - 2 \hat{c}_{i, j - \xi})l_{i + 1, j - \xi} + l_{i + 1, j}\text{,}\text{ otherwise,}\\
    \end{cases}
\end{split}
%}
\end{equation}
where $i\in \left\{ 1, 2, ..., m + 1\right\}$ is the index of stages, $\xi=2^{i-1}$, $j\in \left\{ 1, 2, ..., n\right\}$ is the bit index, $l_{i, j}$ is the \gls{llr} of bit $j$ in stage $i$, and $\hat{c}_{i, j}$ is the hard decision of bit $j$ in stage $i$;
Then, hard decisions from stage-$1$ to stage-$m + 1$ are recursively updated~\cite{PolarDecodersReview}
\begin{equation}
    \hat{c}_{i + 1, j}=
    \begin{cases}
        \hat{c}_{i, j} \oplus \hat{c}_{i, j + 2^{i - 1}}\text{,} & \text{if }\floor{\frac{j - 1}{2^{i - 1}}}\text{mod }2 = 0\text{,}\\
        \hat{c}_{i, j}\text{,} & \text{otherwise}\text{.}\\
    \end{cases}
\end{equation}

For sub-trees with specific patterns of information bits and frozen bits, there exist fast decoding algorithms, so there is no need to traverse down the tree and generate the estimations bit-by-bit.
Commonly used special patterns are rate-$0$ nodes~\cite{SSC,FSC}, rate-$1$ nodes~\cite{SSC,FSC}, \gls{spc} nodes~\cite{FSC}, and repetition nodes~\cite{FSC}.

\subsection{Successive Cancellation List Decoder}
The \gls{scl} decoder retains all possible values (i.e., $0$ and $1$) for a message bit, a list of possible decoding paths is stored, and the final decision is made based on the \gls{pm} for each path. 
To constrain the exponential increase in the list size, the list size of \gls{scl} decoding is limited to a pre-defined number $L$. 
After estimating each bit $u_{i}$, the PM value is updated as follows~\cite{Alexios2015llrscl}
\begin{align*}
     \text{PM}_{i_{p}}=  \sum_{j=0}^{i} \ln\left(1+e^{-(1-2\hat{u}_{j_{p}}) l_{j_{p}}}\right)\text{,}
\end{align*}
where $p$ is the path index, $i$ is the index for message bits $u_{i}$, $l_{j_{p}}$ is the \gls{llr} value of \(u_{j}\) at path $p$, and \( \hat{u}_{j_{p}}\) is the estimate of bit $u_{j}$ at path $p$.
When decoding the \gls{ca} polar codes, the returned estimation is the one that passes the \gls{crc} and has the lowest \gls{pm}.
The \gls{pm} update for \gls{scl} with the special nodes can refer to~\cite{sarkis2016FastSCL,hashemi2017FastSCL}.

\subsection{Perturbation-Enhanced SCL Decoder}
The \gls{pe} \gls{scl} decoder takes received symbols \(y_{i}\) as the input. 
The perturbed version of \(y_{i}\) is denoted as \(y'_{i}\). 
The perturbation noise is denoted as \(\Xi\). 
The relationship between the received symbol, the perturbed version of the received symbol, and the perturbation noise is as follows
\begin{equation}
y'_{i}=y_{i}+\Xi \text{.}
\label{eqn:perturb_soft_information}
\end{equation}

The \gls{pe} \gls{scl} is built upon the \gls{scl} decoder for \gls{ca} polar codes.
At first, a \gls{scl} decoding attempt is performed. If the decoded codeword from the first \gls{scl} decoding attempt passes the CRC, the \gls{pe} \gls{scl} decoder returns the result of the first decoding attempt as the decoded codeword. 
When the \gls{crc} fails, the algorithm enters an iterative process. 
First, the perturbation value \(\Xi\) is generated, and the perturbed received value \(y'_{i}\) is decoded by the SCL decoder. 
If again the result from the current decoding attempt does not pass the \gls{crc}, the \gls{pe} \gls{scl} decoder applies a new perturbation to the received symbol and performs a new decoding attempt.
The \gls{pe} \gls{scl} decoder continues the loop until a pre-defined $T$-th attempt is achieved or a successful \gls{crc} is obtained.

There are two methods to generate the perturbation. 
One method randomly generates the perturbation, and the other method generates a perturbation based on the result of the previous decoding attempt (i.e., adaptive \gls{pe}).
In this work, we call the decoder with the random noise as the \gls{rpe} \gls{scl} decoder, and we call the decoder with the adaptive perturbation as the \gls{pe} \gls{scl} decoder.

The \gls{rpe} decoder uses random noise
\begin{equation}
\label{eq:perturb_value_norm}
  \Xi \sim \mathcal{N}(0,\,\sigma_{p}^{2}),
\end{equation}
as the perturbation, $\sigma_{p}^{2}$ is the variance of the random noise added to the received symbol, and a mean $0$ and variance $\sigma_{p}^{2}$ Gaussian distribution ($\mathcal{N}(0,\,\sigma_{p}^{2})$) is used to generate the random noise.
The adaptive \gls{pe} approach first checks all-agreed and all-disagreed bits.
All-agreed bits are those whose all decoding
paths agree on the bit decisions, including the received hard decision.
All-disagreed bits are those whose all decoding
paths agree on the bit decisions except for the received hard decision.
The remaining bits are called the partially-agreed bits.
The adaptive perturbation generates the perturbation as follows:
\begin{enumerate}
    \item In the $2$nd to $(T-1)$-th attempts, the algorithm applies a biased perturbation
    \begin{equation}
    \label{eqn:perturb_value_bias}
         \Xi=\lambda(-1)^{\hat{c}_i}  \text{,}
    \end{equation}
    where $\lambda=| \sigma_p/\sqrt{2} |$ represents the strength of the perturbation, only to the all-disagreed bits. 
    For all-agreed and partially-agreed bits, a random perturbation described in equation~\eqref{eq:perturb_value_norm} is applied.
    The symbol $\hat{c}_{i}$ denotes the all-agreed and all-disagreed code bits.
    \item In the $T$-th step, which is the last attempt, the algorithm applies a biased perturbation to both all-agreed bits and all-disagreed bits. 
    For partially agreed bits, the random perturbation described in equation~\eqref{eq:perturb_value_norm} is applied.
\end{enumerate}

\section{Methodology}
\label{sec:ablation}
\subsection{Ablation Study}
\label{sec:ablation_study}
In this work, we investigate what factors in the \gls{pe} \gls{scl} decoder contribute to the improvement in the error correction performance, so an ablation study is performed on the adaptive \gls{pe} \gls{scl} decoder.
The parameter $\sigma_{p}$ of the perturbation power  is tuned to a value such that the \gls{rpe} \gls{scl} decoder with a list size of $8$ and $10$ decoding attempts can reach the decoding performance of the \gls{scl} decoder with a list size of $16$.
The parameter $\sigma_{p}$ is in Table~\ref{tab:parameters}, and the prior work~\cite{yang2025biassuccessive} uses similar parameters.

The ablation study is performed using the \gls{awgn} channel, and, to simplify the study, the \gls{bpsk} modulation is used while the \gls{qpsk} modulation is used in~\cite{wang2024adaptive}.
List sizes of $16$ and $8$ are used in this ablation study.
The length of the \gls{crc} bits is $16$ with the polynomial $0\text{x}1021$.
Code lengths $1024$, $4096$, and $8192$ are used in the experiment.
Effective code rates $0.25$, $0.5$, and $0.75$ are selected in the experiment to explore the \gls{pe} decoder under different code rates.
The fast \gls{scl} decoder in~\cite{hashemi2017FastSCL} is used as the constituent decoder in the ablation study, and the number of decoding attempts is set to $T=2$.

From Fig.~\ref{fig:PE_SCL_FER} (a) to (c), the \gls{pe} \gls{scl} decoder with a list size of $8$ (\gls{scl}-$8$ \ref{fig1-4}) has $0.05\text{ dB}$ gap to the \gls{scl} decoder with a list size of $16$ (\gls{scl}-$16$ \ref{fig1-2}) at a \gls{fer} of $10^{-4}$ when $n=1024$.
From Fig.~\ref{fig:PE_SCL_FER} (d) to (i), the \gls{pe} \gls{scl} decoder with a list size of $8$ (\ref{fig1-4}) returns the same decoding performance to the \gls{scl} decoder with a list size of $16$ (\ref{fig1-2}) at a \gls{fer} of $10^{-4}$ when $n=4096$ and $8192$.
This observation is the same as~\cite{wang2024adaptive}, where a larger gain in the decoding performance is returned for codes with a longer code length.
The ordered statistics decoding in the post-processing stage uses a lower order than the ordered statistics decoding alone~\cite{fossorier2002iterative}, so we would like to see whether a low-complexity decoder can be used in the later decoding attempt.
We first test whether the \gls{scl}-$8$ decoder can be replaced by a \gls{sc} decoder in the second decoding attempt for the \gls{pe} \gls{scl} decoder with a list size of $8$.
From Fig.~\ref{fig:PE_SCL_FER} (e), the \gls{pe} decoder with a \gls{sc} decoder (\ref{fig1-5}) in the second decoding attempts (\gls{pe} \gls{scl}-$8$-\gls{sc}) returns a similar decoding performance to the \gls{scl}-$8$ (\ref{fig1-1}) when $n=4096$ and $R=0.5$, and the same trend exists in other $n$s and $R$s.
Hence, a \gls{scl}-$8$ decoder used in the second decoding attempt is necessary for the \gls{pe} decoder to reach the decoding performance of the \gls{scl}-$16$.

We then remove the random noise added to the partially-agreed code bits (\ref{fig1-6}).
From Fig.~\ref{fig:PE_SCL_FER} (e), removing the random noise (\ref{fig1-6}) does not degrade the decoding performance compared to \gls{pe} \gls{scl} decoder (\ref{fig1-4}) when $n=4096$ and $R=0.5$, and the same trend is observed from other $n$s and $R$s.
Hence, the improvement in the error correction performance is mainly from the bias added to the all-agreed and all-disagreed code bits.
It can be seen that the fixed amount of perturbation added to all-agreed and all-disagreed symbols biases the vector of the received symbol toward a valid polar codeword, and the subsequent decoding attempt has a higher chance of reaching a solution.

\subsection{Bias-Enhanced Decoder}
\begin{table}[t]
    \scriptsize
    \centering
    \caption{Parameters $\sigma_{p}$ of the perturbation power for \gls{scl} decoders.}
    \begin{tabular}{@{} c c c c c c @{}}
    \toprule
    & $n=1024$ && $n=4096$ &&  $n=8192$\\
    \midrule
         $R = 0.25$ & $0.25$ && $0.25$ && $0.25$\\
         $R = 0.50$ &  $0.10$ && $0.10$ && $0.10$\\
         $R =0 .75$ &  $0.10$ && $0.10$ && $0.10$\\
    \bottomrule
    \end{tabular}
    \label{tab:parameters}
\end{table}
From the results of Section~\ref{sec:ablation_study}, we can see that the bias derived from the previous decoding attempt and the usage of a \gls{scl} decoder in the subsequent decoding attempt contribute to most of the performance gain in the \gls{pe} \gls{scl} decoder.
Biases added to received symbols of all-agreed code bits further enhance the belief in the received symbols.
Biases added to the received symbol of all-disagreed code bits either correct errors in the received symbols or bias the received symbol toward a code bit of a valid polar codeword.
The benefit of soft bias for all-agreed and all-disagreed code bits is hard to quantify, but those added biases indeed create better soft information, which is closer to a valid polar codeword, for the subsequent decoding attempts, compared to the received symbol used in the initial decoding attempt.

In this work, we proposed to remove the random noise in the \gls{pe} \gls{scl} decoder while keeping the subsequent \gls{scl} decoding attempt with the same list size as the initial decoding attempt.
With all these proposed modifications, the hardware design for the BE decoder might be much simpler than the \gls{pe} decoder, as the unit for generating the random noise is removed while maintaining the same error correction performance.
To differentiate our new methodology to enhance the decoding performance from the \gls{pe} \gls{scl} decoder, we name our method the \gls{be} \gls{scl} decoder.
The \gls{fer} results of our proposed \gls{be} \gls{scl} decoder with a list size of eight correspond to the legends ``\gls{pe} \gls{scl}-8 without noise" in Fig.~\ref{fig:PE_SCL_FER}.
The pseudo-code is in Algorithm~\ref{alg:BE_SCL}.
\begin{algorithm}[t]
    \DontPrintSemicolon
    \SetAlgoVlined
    \footnotesize
    \caption{\label{alg:BE_SCL} \gls{be} \gls{scl} decoding}
    \SetKwFunction{scl}{SCL}
    \SetKwFunction{crc}{CRC}
    %\KwIn{Received soft information vector $\bm{y}$, Perturbation parameter $\sigma_{p}$, The number of decoding attempts $T$, The list size $L$}
    \KwIn{$\bm{y}$, $\sigma_{p}$, $T$, $L$}
    \KwOut{$\hat{\bm{c}}$}
    $\hat{\bm{c}}\leftarrow\scl\left(\bm{y}, L\right)$\;
    \If{$\crc\left(\hat{\bm{c}}\bm{G}^{\top}\right)=\text{False}$}
    {
        \For{$i=2:T$}
        {
          $\left\{\Xi\right\}\leftarrow\left\{\eqref{eqn:perturb_value_bias}|\forall\text{ }\text{all-disagreed bit positions.}\right\}$\;
          %$y'_{i}\leftarrow\eqref{eqn:perturb_soft_information}\forall\text{ }i\in \text{all-disagreed bit positions.}$\;
          \If{$i=T$}
          {
              $\left\{\Xi\right\}\leftarrow\left\{\Xi\right\}\bigcup\left\{\eqref{eqn:perturb_value_bias}|\forall\text{ }\text{all-agreed bit positions.}\right\}$\;
          }
          $\hat{\bm{c}}\leftarrow\scl\left(\bm{y}^{'}\leftarrow\eqref{eqn:perturb_soft_information}, L\right)$\;
          \If{$\crc\left(\hat{\bm{c}}\bm{G}^{\top}\right)=\text{True}$}
          {
              $\text{break}$\;
          }
        }
        
    }
    \KwRet{$\hat{\bm{c}}$}
\end{algorithm}

\begin{figure*}[t]
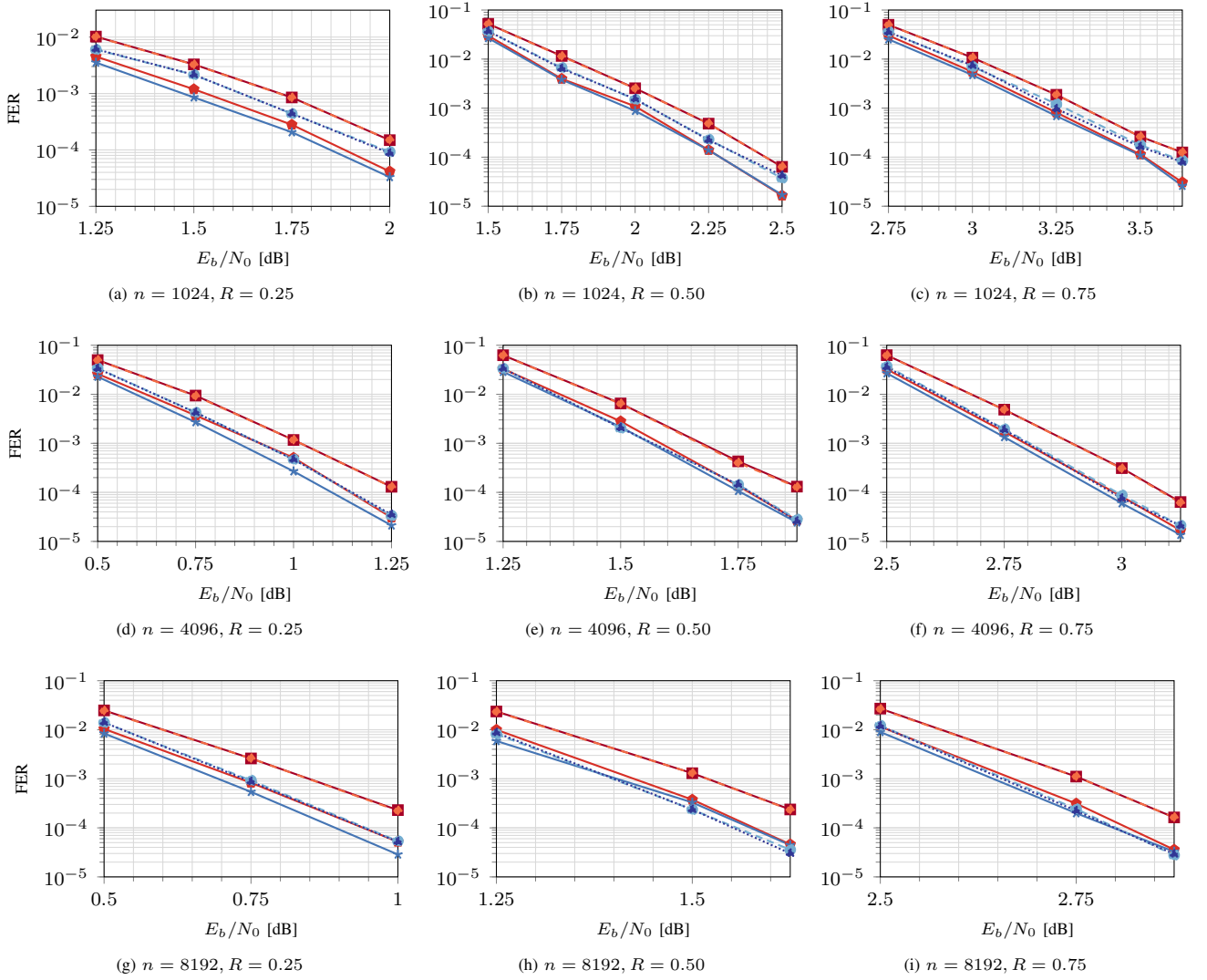

\centering

}
%}
\caption{FERs of decoding polar codes using the perturbation-enhanced SCL decoder, and the bias-enhanced SCL and SCL-SC decoders.}
\label{fig:PE_SCL_FER}
\end{figure*}

\begin{figure}[h!]
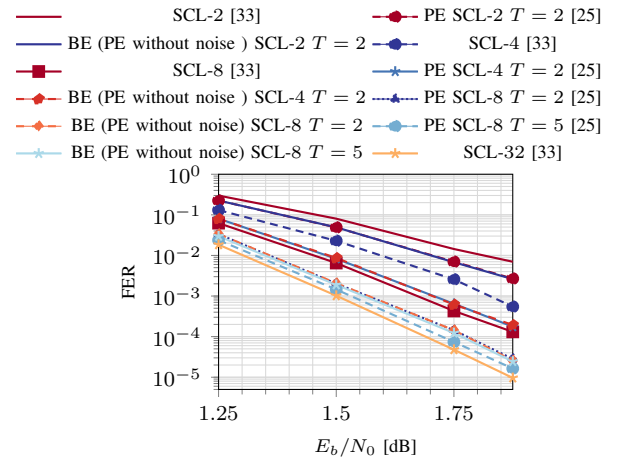

\centering
%

\caption{FERs of decoding polar codes using the \gls{pe} SCL decoder with different parameters when decoding the $n=4096$ and $R=0.5$ polar code.}
\label{fig:FER_diff_parameters}
\end{figure}
\section{Further Discussions}
\label{sec:results}
\textit{Different Numbers of Decoding Attempts:} We also investigate the effect of setting the number of decoding attempts $T>2$, and conduct the experiments on $n=4096$ and $R = 0.5$ polar codes.
As we can see from Fig.~\ref{fig:FER_diff_parameters}, with a list size of $8$, the \gls{be} \gls{scl} decoder has negligible gain in the error correction performance when increasing the number of decoding attempts from $2$ to $5$ while a gain of $0.1$ dB in the error correction performance is observed from the \gls{pe} \gls{scl} decoder when increasing the number of decoding attempts to $T=5$.
Hence, we can conclude that the random noise is the key to exploring new results for \gls{pe} decoders.
The \gls{fer} of \gls{pe} \gls{scl} decoder with $T=5$ is lower bounded by the \gls{scl} decoder with a list size of $32$, which means the gain in error correction performance is less significant after $T=2$.
At a target \gls{fer} of $10^{-3}$, by increasing $T$ from $1$ to $2$, the \gls{pe} \gls{scl} decoder improves the \gls{fer} by $0.1\text{ dB}$, while the gain of increasing $T$ from $2$ to $5$ is less than $0.05\text{ dB}$ in the \gls{fer}.
Given this marginal gain, we can conclude that most of the hard-to-decode frames are fixed by the bias (not by improving the \gls{scl} decoding) to improve the decoding performance when $T=2$, small improved decoding performance is seen for $T>2$, and the hard-to-decode frames become increasingly rare as the $E_{b}/N_{0}$ increases.
We conclude that the gain in the decoding performance is marginal when $T>2$.
Hence, we will focus on the case of $T=2$.

\textit{Different List Sizes:} We test the \gls{be} \gls{scl} decoder on polar codes with $n=4096$ and $R=0.5$, and a small list size of $4$ is used.
From Fig~\ref{fig:FER_diff_parameters}, we can see that the decoding performance of the \gls{be} \gls{scl} decoder is consistent with the decoding performance of the \gls{pe} \gls{scl} decoder when a list size of $4$ is used.
For a list size of $2$, the decoding performance of the \gls{be} \gls{scl} decoder is consistent with the decoding performance of the \gls{pe} \gls{scl} decoder, but these enhancement-based decoders have a $0.1$ dB loss at a  \gls{fer} of $10^{-2}$ compared to the \gls{scl} decoder with a list size of $4$.

\textit{Complexity Analysis:} For the decoding attempt $T=2$, both the \gls{pe} and \gls{be} \gls{scl} decoder invoke the second decoding attempt when the initial \gls{scl} decoder fails to find a result that can pass the \gls{crc}.
The chance of receiving symbols that will result in a decoding failure in the \gls{scl} decoder will be similar, given the similar simulation environment for both the \gls{pe} and the \gls{be} \gls{scl} decoders, so the average number of decoding attempts of the \gls{be} \gls{scl} decoder will be equal to the \gls{pe} \gls{scl} decoder when $T=2$.
The average number of decoding attempts (Avg. $T$) is shown in Fig.~\ref{fig:PE_SCL_Avg_T}, and the results confirm our conjecture.
Also, from Fig.~\ref{fig:PE_SCL_Avg_T}, the number of decoding attempts is close to one for codes under the whole $E_{b}/N_{0}$ range.
This reduced number of decoding attempts implies that, in a fixed-latency implementation for the \gls{be} \gls{scl} decoder, the decoder can be shut down for most of the time, so the second decoding attempt has minimal contribution to the energy consumption.

\begin{figure*}[ht]
\centering

}
%}
\caption{The average number of decoding attempts for PE and BE SCL decoders with $T=2$ and $L=8$.}
\label{fig:PE_SCL_Avg_T}
\end{figure*}
\section{Algorithm with Reduced Memory Usage}
\label{sec:redcuedmem}

\begin{figure}[t]
    \centering
    \includegraphics[scale=0.29]{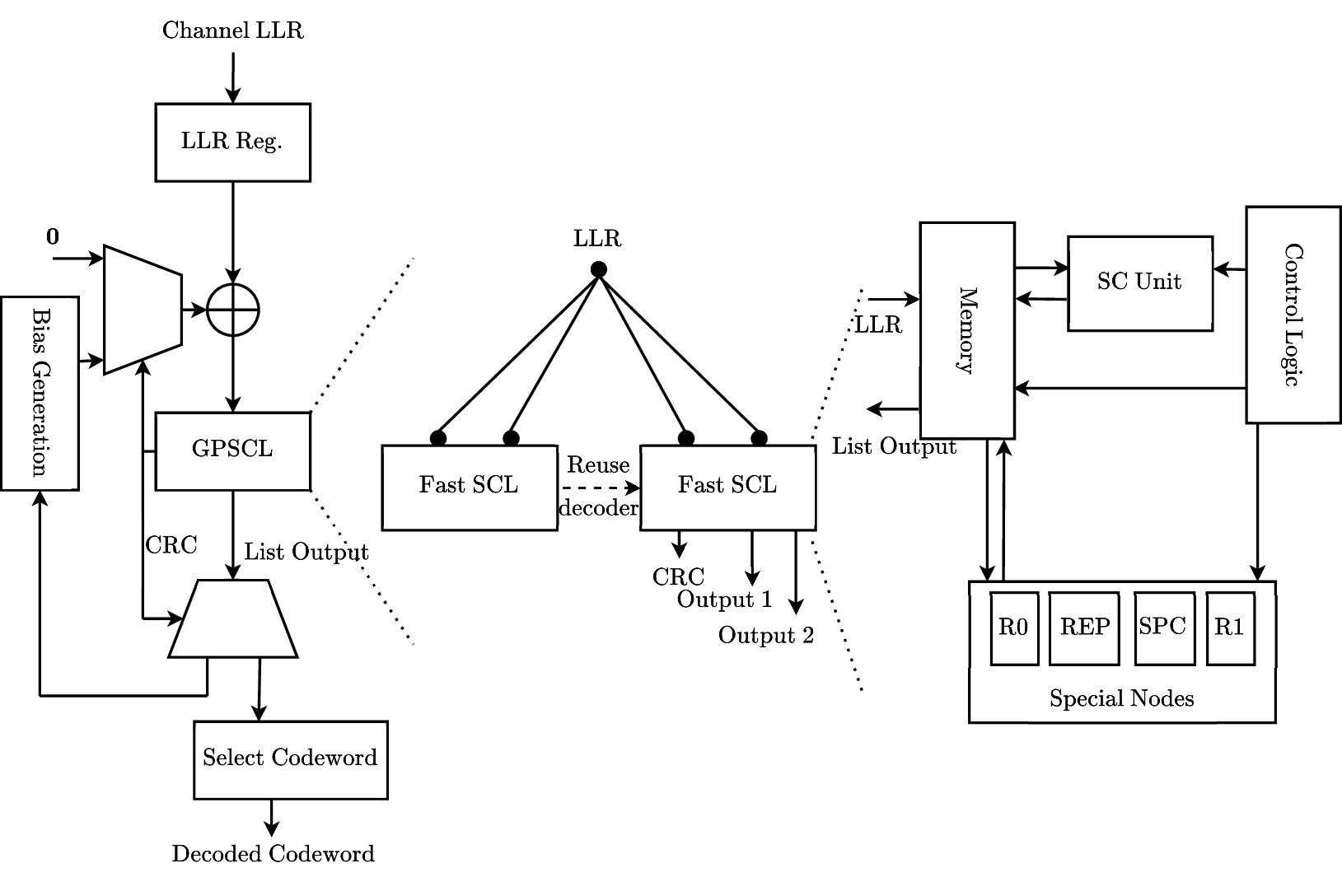}
    \caption{The block diagram of the BE GPSCL decoder with the fast \gls{scl} component decoder and $S=2$.}
    \label{fig:block_diagram_BE_GPSCL}
\end{figure}

\begin{figure*}
\centering
%
}
%}
\caption{FERs of decoding polar codes using the BE GPSCL decoders and the SCL decoder.}
\label{fig:partition_BE_SCL}
\end{figure*}

\subsection{BE GPSCL Decoders for Long Polar Codes}

The \gls{scl} decoder has a large memory usage of $O\left(Ln\right)$~\cite{SCL} where $n-1$ memory blocks are required to store the intermediate value in each path of the \gls{scl} decoder, and there are $L$ paths in the \gls{scl} decoder.
The large area requirement of the \gls{scl} decoder is mostly dominated by memory blocks used in the implementation~\cite{Alexios2015llrscl}.
Hence, the reduced list size means \gls{pe} and \gls{be} \gls{scl} decoders have area-efficient implementations for polar codes compared to the conventional \gls{scl} decoder.

In this work, we target the usage of the \gls{be} \gls{scl} decoder for long polar codes where a large memory usage is required due to the large code length $n$.
To reduce the memory usage of the \gls{be} \gls{scl} decoder, we propose to use the \gls{gpscl} decoder~\cite{Hashemi2018partitionlist} as the component decoder for the \gls{be} \gls{scl} decoder in this work.
The partitioned \gls{scl} decoder performs list decoding only on the lower level of the decoding tree that has a shorter code length $n$, while the \gls{sc} decoding is performed on the higher level of the decoding tree~\cite{Hashemi2016partitionlist,Hashemi2017partitionlist}.
Hence, only one path is returned on the higher level of the decoding tree, and the memory usage is reduced compared to returning $L$ paths on the higher level of the decoding tree.
In the partitioned decoder, a parameter for controlling the decoding performance and the memory trade-off is the number of paths ($S$) returned in the higher level~\cite{Hashemi2018partitionlist}, and the partitioned decoder that uses this parameter is called the \gls{gpscl} decoder.
In this work, a \gls{gpscl} decoder with a partition ($P$) of two and $S=2$ is used.
Our target decoding performance can be achieved when $S=2$, and we choose to use $S=2$ in our work because the smaller the $S$, the larger the reduction in the memory complexity, hence the larger area reduction~\cite{Hashemi2018partitionlist}.

The \gls{be} \gls{gpscl} works as the following:
\begin{itemize}
    \item The \gls{gpscl} decoder performs the \gls{scl} decoding on the shorter polar sub-codes.
    In this work, we use the fast \gls{scl} decoder~\cite{hashemi2017FastSCL}.
    \item For \gls{gpscl} decoders, candidate codewords, which pass the \gls{crc}, are returned. 
    If the number of candidate codewords is fewer than $S$, codewords, which fail the \gls{crc}, with a relatively smaller \gls{pm} are returned. 
    \item The \gls{crc} of the last partition is used to activate the bias-enhancement instead of checking \gls{crc} for all partitions, which removes the need to save the \gls{crc} results of all previous partitions.
    If this \gls{crc} passes, the \gls{be} \gls{gpscl} decoder considers that a correct decoded codeword is generated, and the decoding stops.
    If not, the bias enhancement for the next decoding attempt is generated based on the $S$ returned candidate codewords, unlike the $L$ candidate codewords used by the \gls{be} \gls{scl} decoder.
\end{itemize}
Fig~\ref{fig:block_diagram_BE_GPSCL} illustrates our \gls{be} \gls{gpscl} decoder.
Experiments are only performed for polar codes with a length of $4096$ and $8192$ because we are interested in the polar codes with a code length that is larger than the code length used in the 5G communication standard~\cite{embb} (i.e., $> 1024$), and is smaller or equal to the $8$K code lengths ($n=8448$) support by the control and the data channel used by the current 5G standard~\cite{embb}.

Fig.~\ref{fig:partition_BE_SCL} shows the simulation results of the \gls{be} \gls{gpscl} decoder.
The parameters for the perturbation power of the \gls{be} \gls{gpscl} decoder are the same as the parameters shown in Table~\ref{tab:parameters}.
As suggested by~\cite{Hashemi2018partitionlist}, we use \gls{crc} codes with different lengths for different partitions,
To maintain the same effective code rate, a \gls{crc} code with a length of $6$ is used for the first partition, and a \gls{crc} code with a length of $10$ is used for the second partition.
For the length-$6$ \gls{crc} code, the \gls{crc} polynomial is $0\text{x} 21$, while, for the length-$10$ \gls{crc} code, the \gls{crc} polynomial is $0\text{x} 233$.
From Fig.~\ref{fig:partition_BE_SCL} (b), \gls{be} \gls{gpscl}, with a list size of $8$ and $S=2$, can return the same error correction performance as the performance of a \gls{scl} decoder with a list size of $16$ when $n=4096$ and $R=0.5$.
The same trend exists for all other code lengths and rates.
For \gls{be} \gls{gpscl} decoder, with a list size of $8$ and $S=1$, it cannot achieve the error correction performance of the \gls{scl} decoder with a list size of $16$ when $n=4096$ and $R=0.50$.
This degradation in the error correction performance can be observed from all other code rates and lengths.

Also, from Fig.~\ref{fig:partition_BE_SCL} (b), the \gls{gpscl} decoder with $S=2$ and a list size of $8$ has a similar error correction performance to the \gls{scl} with a list size of $8$ when $n=4096$ and $R=0.50$, while the \gls{be} \gls{gpscl} decoder with $S=2$ and a list size of $8$ has a similar error correction performance to the \gls{scl} with a list size of $16$.
The bias used by the \gls{be} decoder is generated according to the two decoded codewords in the list.
Hence, we can conclude that only two candidate codewords are enough to generate a biased enhancement to improve the error correction performance.
The number of codewords is not the key factor to generate accurate biased enhancement, and the accuracy of the codeword estimation is the key factor to generate an accurate biased enhancement.
Moreover, compared to the \gls{be} \gls{scl} decoder, the determination of all-agree and all-disagree positions takes fewer operations as the \gls{gpscl} decoder has fewer codewords in the list.

When $T=2$, the bias is applied to all-agreed and all-disagreed bits in the \gls{be} \gls{scl} decoder.
The XOR gate and the priority encoder can be used to check all-agreed and all-disagreed bits for the \gls{be} \gls{scl} decoder.
Bits in the first codeword in the list are XORed with the corresponding bits in all other codewords in the list, $L-1$ XOR gates are required, the result is a length-$\left(L-1\right)$ binary vector, and a bit $1$ is appended to the index $1$ of the length-$\left(L-1\right)$ binary vector.
The $L$-to-$\log_{2}\left(L\right)$ priority encoder processes the length-$L$ binary output from the XOR gate array.
An output that is equal to $0$ from the priority encoder indicates an all-agreed or an all-disagreed bit because this input implies all bits in the list agree.
In total, $\left(L-1\right)n$ XOR gates and $n$ priority encoders are used to determine the all-agreed and all-disagreed bits for the \gls{be} \gls{scl} decoder.

When $T=2$, the bias is also applied to all-agreed and all-disagreed bits in the \gls{be} \gls{gpscl} decoder.
When the \gls{gpscl} decoder with $S=2$ is used, only one XOR gate is required to determine the all-agreed and all-disagreed bits.
An XOR gate is used to compare the bits in the first returned path and the second returned path, and an output $0$ means two bits agree, which implies an all-agreed or an all-disagreed bit, and an output $1$ returns otherwise.
In total, $n$ XOR gates are used to determine the all-agreed and all-disagreed bits for the \gls{be} \gls{gpscl} decoder, which is a smaller overhead than the $\left(L-1\right)n$ XOR gates plus $n$ $L$-to-$\log_{2}\left(L\right)$ priority encoders.

\subsection{Quantized Model of the \gls{be} \gls{gpscl} Decoder}
Fig.~\ref{fig:partition_BE_SCL_quantized} shows the decoding performance of the quantized model.
Table~\ref{tab:bit-width} shows the bit-widths used for the quantized model, where $q_{i}$ is the bit-width for the integer part and $q_{f}$ is the bit-width for the fractional part.
We compare our \gls{be} \gls{gpscl} decoder with the fast \gls{scl} decoder~\cite{hashemi2017FastSCL}.
For the received \gls{llr} and the internal \gls{llr}, a sign bit is included in the integer part, and the sign and magnitude representation is used.
The \gls{pm} is unsigned, and all bits are used to represent the magnitude.
Fig.~\ref{fig:partition_BE_SCL_quantized} shows the \gls{fer} of the quantized \gls{be} \gls{gpscl} decoder with the fast \gls{scl} component decoder.

The memory usage, which is measured by the number of bits, of the \gls{scl} and the \gls{gpscl} can be  computed as follows~\cite{Hashemi2018partitionlist}:
\begin{align*}
    M_{\text{SCL}} = &n * Q_{\text{LLR}} + \left(n-1\right)*L*Q_{\alpha} + L*Q_{\text{PM}}\\
    &+ \left(2n-1\right)*L\text{;}
\end{align*}
\begin{align*}
    M_{\text{GPSCL}} = &n * Q_{\text{LLR}} + \left(S * \sum_{i=1}^{\log_{2}(P)} \frac{n}{2^{i}} + L \left(\frac{n}{P} - 1\right)\right)Q_{\alpha}\\
    &+ L*Q_{\text{PM}} + S * \sum_{i=1}^{\log_{2}(P)} \frac{n}{2^{i}} + L \left(\frac{2n}{P} - 1\right)\text{.}
\end{align*}
The bit-width for the channel \gls{llr}, the internal \gls{llr}, and the \gls{pm} are denoted as $Q_{\text{LLR}}$, $Q_{\alpha}$, and $Q_{\text{PM}}$ respectively.
The bit-width is shown in Table~\ref{tab:bit-width}.
The component decoder for the \gls{be} \gls{scl} decoder is the \gls{scl} decoder, and the memory usage of the \gls{be} \gls{scl} decoder is equal to the memory usage of the \gls{scl} decoder.
The component decoder for the \gls{be} \gls{gpscl} decoder is the \gls{gpscl} decoder, and the memory usage of the \gls{be} \gls{gpscl} decoder is equal to the memory usage of the \gls{gpscl} decoder.

From Fig.~\ref{fig:partition_BE_SCL_quantized} (b), the quantized model of the \gls{be} \gls{gpscl} decoder has a loss of less than $0.05\text{ dB}$ at a \gls{fer} of $2\times10^{-4}$ when decoding polar codes with $n=8192$ and $R=0.25$.
For all other code lengths and rates, the loss due to the quantization is also less than $0.05\text{ dB}$.
According to Table~\ref{tab:bit-width}, the \gls{be} \gls{gpscl} decoder can be quantized using the same bit-width as the \gls{scl} decoder.
The memory usage of the proposed techniques is shown in Table~\ref{tab:mem_com}.
Compared to the memory usage of the \gls{scl} decoder with a list size of $16$, the \gls{be} \gls{scl} decoder reduces the memory usage by $48\%$.
Compared to the memory usage of the \gls{scl} decoder with a list size of $16$, the memory reduction of the \gls{be} \gls{gpscl} decoder is $67\%$.

\begin{table}[t]
    \scriptsize
    \centering
    \caption{Quantization bit-widths for Decoders.}
    \begin{tabular}{@{} c c cc c cc c cc@{}}
    \toprule
    &\multicolumn{2}{c}{Received \gls{llr}}&\phantom{a}&\multicolumn{2}{c}{Internal \gls{llr}}&\phantom{a}&\multicolumn{2}{c}{\gls{pm}}\\
    \cmidrule{2-3} \cmidrule{5-6} \cmidrule{8-9}
    &$q_{i}$ & $q_{f}$ && $q_{i}$ & $q_{f}$ && $q_{i}$ & $q_{f}$ \\
    \midrule
         $n=4096$, $R=0.25$ & $4$ & $2$ && $6$ & $2$ && $7$ & $2$\\
         $n=8192$, $R=0.25$ & $4$ & $2$ && $6$ & $2$ && $8$ & $2$\\
         $R=0.50$ & $4$ & $2$ && $6$ & $2$ && $7$ & $2$\\
         $R=0.75$ & $4$ & $2$ && $6$ & $2$ && $7$ & $2$\\
         %\midrule
    \bottomrule
    \end{tabular}
    \label{tab:bit-width}
\end{table}
\begin{table}
    \centering
    \caption{memory usage (Mem.) of different decoders.}
    \begin{tabular}{@{}c c c c c c c c@{}}
    \toprule
    &\multicolumn{1}{c}{\gls{scl}-$16$}&\phantom{a}&\multicolumn{1}{c}{\gls{be} \gls{scl}-$8$}&\phantom{a}&\multicolumn{1}{c}{\gls{be} \gls{gpscl}-$8$}\\
    &Mem. [bits] && Mem. [bits] && Mem. [bits]\\
    \midrule
    $n=4096$ &  $6.80\times10^{5}$ && $3.52\times10^{5}$&& $2.25\times10^{5}$\\
    $n=8192$ & $1.36\times10^{6}$ && $7.05\times10^{5}$&& $4.51\times10^{5}$ \\ 
    \bottomrule
    \end{tabular}
    \label{tab:mem_com}
\end{table}

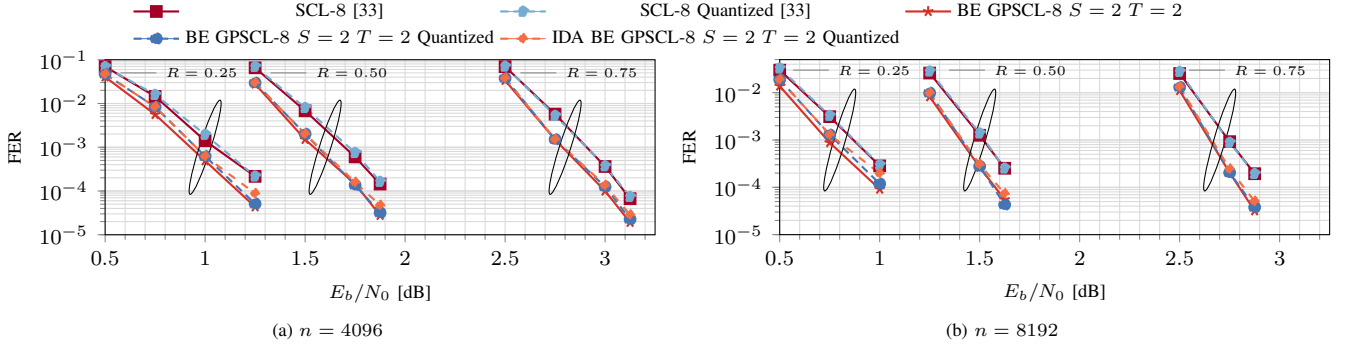
\begin{figure*}[t]
\centering
\begin{tikzpicture}
  \begin{customlegend}[legend columns=3,legend style={at={(0.0,-2.0)},anchor=south west, align=left,draw=none},
  legend entries={
  \scriptsize{SCL-$8$~\cite{hashemi2017FastSCL}}, \scriptsize{SCL-$8$ Quantized~\cite{hashemi2017FastSCL}}, \scriptsize{BE GPSCL-$8$ $S=2$ $T=2$}, \scriptsize{BE GPSCL-$8$ $S=2$ $T=2$ Quantized}, \scriptsize{IDA BE GPSCL-$8$ $S=2$ $T=2$ Quantized}
  }
  ]
  \addlegendimage{draw=colorblindfree11_1, mark=square*, fill=colorblindfree11_1}
  
  \addlegendimage{draw=colorblindfree11_9,mark=pentagon*,fill=colorblindfree11_9, densely dashed}
  \addlegendimage{draw=colorblindfree11_2,mark=star,fill=colorblindfree11_2}
  
  \addlegendimage{draw=colorblindfree11_10,mark=*,fill=colorblindfree11_10, densely dashed}
  \addlegendimage{draw=colorblindfree11_3,mark=diamond*,fill=colorblindfree11_3, dashed}
  
  %\addlegendimage{draw=colorblindfree11_11,mark=triangle*,fill=colorblindfree11_11, densely dotted}

  %\addlegendimage{draw=colorblindfree11_8,mark=square*,fill=colorblindfree11_8}
  %\addlegendimage{draw=colorblindfree11_4,mark=star,fill=colorblindfree11_4}

  %\addlegendimage{draw=colorblindfree11_1,mark=*,fill=colorblindfree11_1}
  %\addlegendimage{draw=colorblindfree11_10,mark=triangle*,fill=colorblindfree11_10}

  %\addlegendimage{draw=colorblindfree11_2,mark=triangle*,fill=colorblindfree11_2}

  %\addlegendimage{draw=colorblindfree11_11,fill=colorblindfree11_11}
  \end{customlegend}
\end{tikzpicture}
\vspace{-1.6em}

\captionsetup[subfigure]{font=scriptsize,labelfont=scriptsize}
%\resizebox{.8\textwidth}{!}{
\subfloat[\scriptsize{$n=4096$}]{
    \centering
    \begin{tikzpicture}[spy using outlines={circle,black,magnification=2, connect spies}]
      \begin{axis}[
        footnotesize, width=\columnwidth, height=0.44\columnwidth,     
        xmin=0.5, xmax=3.25, xtick={0.50,1.0,...,3.25},minor x tick num=4,
        ymin=1e-5,  ymax=1e-1,
        xlabel=\scriptsize{$E_b/N_0 \text{ [dB]}$}, ylabel=\scriptsize{FER},  ylabel near ticks,
        grid=both, grid style={gray!30}, ymode=log, log basis y={10},
        tick align=outside, tickpos=left, legend columns=1,legend style={
        at={(0.50,1.50)},anchor=center},
        cycle list = {
        {draw=colorblindfree11_1,mark=square*, mark options = {fill=colorblindfree11_1}},
        {draw=colorblindfree11_9,mark=pentagon*,mark options ={fill=colorblindfree11_9}, densely dashed},
        {draw=colorblindfree11_2,mark=star,mark options = {fill=colorblindfree11_1}},
        {draw=colorblindfree11_10, mark = *, mark options = {fill= colorblindfree11_10}, densely dashed},
        {draw=colorblindfree11_3, mark = diamond*, mark options = {fill= colorblindfree11_3}, dashed}
        %{draw=colorblindfree11_11,mark=triangle*,mark options = {fill=colorblindfree11_11}, densely dotted}
        }
        ]
        \addplot table [x=Eb/N0, y=SCL8DisCRC] {data/n_4096_R_0_25_CRC_16_dist_CRC.txt};
        \addplot table [x=Eb/N0, y=SCL8DisCRCQuant] {data/n_4096_R_0_25_CRC_16_dist_CRC.txt};
        \addplot table [x=Eb/N0, y=BEGPSCL8L2DisCRC] {data/n_4096_R_0_25_CRC_16_dist_CRC.txt};
        \addplot table [x=Eb/N0, y=BEGPSCL8L2DisCRCQuant] {data/n_4096_R_0_25_CRC_16_dist_CRC.txt};
        \addplot table [x=Eb/N0, y=IDABEGPSCL8L2Quant] {data/n_4096_R_0_25_CRC_16_dist_CRC.txt};
        \node [pin=0:\tiny{$R=0.25$}] at (0.53,5e-2) {};
        \draw [black] (1.0,1e-3) ellipse [rotate=0, x radius=0.08, y radius=0.40,];
        \addplot table [x=Eb/N0, y=SCL8DisCRC] {data/n_4096_R_0_5_CRC_16_dist_CRC.txt};
        \addplot table [x=Eb/N0, y=SCL8DisCRCQuant] {data/n_4096_R_0_5_CRC_16_dist_CRC.txt};
        \addplot table [x=Eb/N0, y=BEGPSCL8L2DisCRC] {data/n_4096_R_0_5_CRC_16_dist_CRC.txt};
        \addplot table [x=Eb/N0, y=BEGPSCL8L2DisCRCQuant] {data/n_4096_R_0_5_CRC_16_dist_CRC.txt};
        \addplot table [x=Eb/N0, y=IDABEGPSCL8L2Quant] {data/n_4096_R_0_5_CRC_16_dist_CRC.txt};
        \node [pin=0:\tiny{$R=0.50$}] at (1.28,5e-2) {};
        \draw [black] (1.60,1e-3) ellipse [rotate=0, x radius=0.08, y radius=0.40,];
        \addplot table [x=Eb/N0, y=SCL8DisCRC] {data/n_4096_R_0_75_CRC_16_dist_CRC.txt};
        \addplot table [x=Eb/N0, y=SCL8DisCRCQuant] {data/n_4096_R_0_75_CRC_16_dist_CRC.txt};
        \addplot table [x=Eb/N0, y=BEGPSCL8L2DisCRC] {data/n_4096_R_0_75_CRC_16_dist_CRC.txt};
        \addplot table [x=Eb/N0, y=BEGPSCL8L2DisCRCQuant] {data/n_4096_R_0_75_CRC_16_dist_CRC.txt};
        \addplot table [x=Eb/N0, y=IDABEGPSCL8L2Quant] {data/n_4096_R_0_75_CRC_16_dist_CRC.txt};
        \node [pin=0:\tiny{$R=0.75$}] at (2.53,5e-2) {};
        \draw [black] (2.80,1e-3) ellipse [rotate=0, x radius=0.08, y radius=0.40,];
      \end{axis}
    \end{tikzpicture}
}
\subfloat[\scriptsize{$n=8192$}]{
    \centering
    \begin{tikzpicture}[spy using outlines={circle,black,magnification=2, connect spies}]
      \begin{axis}[
        footnotesize, width=\columnwidth, height=0.44\columnwidth,     
        xmin=0.5, xmax=3.25, xtick={0.5,1.0,...,3.25},minor x tick num=4,
        ymin=1e-5,  ymax=5e-2,
        xlabel=\scriptsize{$E_b/N_0 \text{ [dB]}$}, 
        ylabel=\scriptsize{FER},  
        ylabel near ticks,
        grid=both, grid style={gray!30}, ymode=log, log basis y={10},
        tick align=outside, tickpos=left, legend columns=1,legend style={
        at={(0.50,1.50)},anchor=center},
        cycle list = {
        {draw=colorblindfree11_1,mark=square*, mark options = {fill=colorblindfree11_1}},
        {draw=colorblindfree11_9,mark=pentagon*,mark options ={fill=colorblindfree11_9}, densely dashed},
        {draw=colorblindfree11_2,mark=star,mark options = {fill=colorblindfree11_2}},
        {draw=colorblindfree11_10, mark = *, mark options = {fill= colorblindfree11_10}, densely dashed},
        {draw=colorblindfree11_3, mark = diamond*, mark options = {fill= colorblindfree11_3}, dashed}
        %{draw=colorblindfree11_11,mark=triangle*,mark options = {fill=colorblindfree11_11}, densely dotted}
        }
        ]
        \addplot table [skip coords between index={3}{4},x=Eb/N0, y=SCL8DisCRC] {data/n_8192_R_0_25_CRC_16_dist_CRC.txt};
        \addplot table [skip coords between index={3}{4},x=Eb/N0, y=SCL8DisCRCQuant] {data/n_8192_R_0_25_CRC_16_dist_CRC.txt};
        \addplot table [skip coords between index={3}{4},x=Eb/N0, y=BEGPSCL8L2DisCRC] {data/n_8192_R_0_25_CRC_16_dist_CRC.txt};
        \addplot table [skip coords between index={3}{4},x=Eb/N0, y=BEGPSCL8L2DisCRCQuant] {data/n_8192_R_0_25_CRC_16_dist_CRC.txt};
        \addplot table [skip coords between index={3}{4},x=Eb/N0, y=IDABEGPSCL8L2Quant] {data/n_8192_R_0_25_CRC_16_dist_CRC.txt};
        \node [pin=0:\tiny{$R=0.25$}] at (0.52,3e-2) {};
        \draw [black] (0.8,1e-3) ellipse [rotate=0, x radius=0.08, y radius=0.40,];
        \addplot table [skip coords between index={3}{4},x=Eb/N0, y=SCL8DisCRC] {data/n_8192_R_0_5_CRC_16_dist_CRC.txt};
        \addplot table [skip coords between index={3}{4},x=Eb/N0, y=SCL8DisCRCQuant] {data/n_8192_R_0_5_CRC_16_dist_CRC.txt};
        \addplot table [skip coords between index={3}{4},x=Eb/N0, y=BEGPSCL8L2DisCRC] {data/n_8192_R_0_5_CRC_16_dist_CRC.txt};
        \addplot table [skip coords between index={3}{4},x=Eb/N0, y=BEGPSCL8L2DisCRCQuant] {data/n_8192_R_0_5_CRC_16_dist_CRC.txt};
        \addplot table [skip coords between index={3}{4},x=Eb/N0, y=IDABEGPSCL8L2Quant] {data/n_8192_R_0_5_CRC_16_dist_CRC.txt};
        \node [pin=0:\tiny{$R=0.50$}] at (1.30,3e-2) {};
        \draw [black] (1.5,1e-3) ellipse [rotate=0, x radius=0.08, y radius=0.40,];
        \addplot table [skip coords between index={3}{4},x=Eb/N0, y=SCL8DisCRC] {data/n_8192_R_0_75_CRC_16_dist_CRC.txt};
        \addplot table [skip coords between index={3}{4},x=Eb/N0, y=SCL8DisCRCQuant] {data/n_8192_R_0_75_CRC_16_dist_CRC.txt};
        \addplot table [skip coords between index={3}{4},x=Eb/N0, y=BEGPSCL8L2DisCRC] {data/n_8192_R_0_75_CRC_16_dist_CRC.txt};
        \addplot table [skip coords between index={3}{4},x=Eb/N0, y=BEGPSCL8L2DisCRCQuant] {data/n_8192_R_0_75_CRC_16_dist_CRC.txt};
        \addplot table [skip coords between index={3}{4},x=Eb/N0, y=IDABEGPSCL8L2Quant] {data/n_8192_R_0_75_CRC_16_dist_CRC.txt};
        \node [pin=0:\tiny{$R=0.75$}] at (2.50,3e-2) {};
        \draw [black] (2.7,1e-3) ellipse [rotate=0, x radius=0.08, y radius=0.40,];
      \end{axis}
    \end{tikzpicture}
}
\caption{FERs of decoding polar codes using the quantized BE GPSCL decoders and the SCL decoder.}
\label{fig:partition_BE_SCL_quantized}
\end{figure*}

\section{Algorithm with Reduced Computational Complexity}
\label{sec:reducedcomp}
It is shown in~\cite{johannsen2025dpaSCL} that, by adaptively deactivating the decoding path in the \gls{scl} decoder, parts of hardware components are disabled via clock gating, and the power consumption of the \gls{scl} decoder is reduced because the energy consumption of the \gls{scl} decoder is dominated by the switching power of the sequential circuits~\cite{tao2021powerbreakdownscl}.
In this section, we propose to adaptively select the list size for the decoder.
It is unknown whether accurate bias can be generated under the adaptive list size.
We find that accurate bias enhancements can still be generated after adopting \gls{ida} decoding to adaptively select the list size in the first decoding attempt, while having little degraded decoding performance.

We propose to adaptively choose the list size based on the received soft information for the first decoding attempt.
Adaptive selection of list sizes preserves the error correction performance of the \gls{scl} decoder used in the first decoding attempt while reducing the computational complexity.
Compared to the \gls{be} \gls{gpscl} decoder with a list size of $8$, Fig.~\ref{fig:FER_diff_IDA} shows that using a \gls{be} decoder with a fixed list size of $4$ in the first decoding attempt causes $0.05\text{ dB}$ degradation in the error correction performance at a \gls{fer} of $10^{-4}$ compared to using a list size of $8$ in the first decoding attempt.

The \gls{ida} decoding adaptively determines the list size as follows~\cite{IDASCL}.
All received \gls{llr}s are compared with the \gls{llr} threshold $\gamma$, and \gls{llr}s whose magnitude is smaller than or equal to $\gamma$ are considered as unreliable.
If the number of unreliable \gls{llr} is smaller than the threshold $\varphi$, a small list size is used.
Otherwise, a large list size is used.
As the channel condition is getting better, the average list size used by the decoding is reduced, as the small list size is used more often than the large list size.
In this work, the \gls{ida} decoding~\cite{IDASCL} is applied on the first decoding attempt of the \gls{be} \gls{gpscl} decoder.

\begin{figure}
\centering
\begin{tikzpicture}
  \begin{customlegend}[legend columns=1,legend style={at={(0.0,-2.0)},anchor=south west, align=left,draw=none},
  legend entries={
  \scriptsize{BE GPSCL-$8$}, \scriptsize{BE GPSCL-$8$ with IDA decoding on the first decoding attempt}, \scriptsize{BE GPSCL-$8$ with IDA decoding on the first and second decoding attempt}, \scriptsize{$L=4$ for the first decoding and $L=8$ for the second decoding}
  }
  ]
  \addlegendimage{draw=colorblindfree11_1, mark=square*, fill=colorblindfree11_1}
  \addlegendimage{draw=colorblindfree11_10,mark=star,fill=colorblindfree11_10}
  
  \addlegendimage{draw=colorblindfree11_2,mark=pentagon*,fill=colorblindfree11_2, dashed}
  \addlegendimage{draw=colorblindfree11_11,mark=triangle*,fill=colorblindfree11_11, densely dotted}
  
  \addlegendimage{draw=colorblindfree11_3,mark=diamond*,fill=colorblindfree11_3, dashed}
  \addlegendimage{draw=colorblindfree11_9,mark=*,fill=colorblindfree11_9, densely dashed}

  \addlegendimage{draw=colorblindfree11_8,mark=star,fill=colorblindfree11_8}
  \addlegendimage{draw=colorblindfree11_4,mark=star,fill=colorblindfree11_4}

  %\addlegendimage{draw=colorblindfree11_11,fill=colorblindfree11_11}
  \end{customlegend}
\end{tikzpicture}
\vspace{-0.6em}

    \begin{tikzpicture}[spy using outlines={circle,black,magnification=2, connect spies}]
      \begin{axis}[
        footnotesize, width=.618\columnwidth, height=.5\columnwidth,     
        xmin=1.25, xmax=1.875, xtick={1.25,1.5,...,2},minor x tick num=4,
        ymin=5e-6,  ymax=1e-1,
        xlabel=\scriptsize{$E_b/N_0 \text{ [dB]}$}, 
        ylabel=\scriptsize{FER},  
        %ylabel near ticks,
        grid=both, grid style={gray!30}, ymode=log, log basis y={10},
        tick align=outside, tickpos=left, legend columns=1,legend style={
        at={(0.50,1.50)},anchor=center},
        cycle list = {
        {draw=colorblindfree11_1,mark=square*, mark options = {fill=colorblindfree11_1}},
        {draw=colorblindfree11_10,mark=star,mark options = {fill=colorblindfree11_10}},
        {draw=colorblindfree11_2,mark=pentagon*,mark options ={fill=colorblindfree11_2}, dashed},
        {draw=colorblindfree11_11,mark=triangle*,mark options = {fill=colorblindfree11_11}, densely dotted},
        {draw=colorblindfree11_3, mark = diamond*, mark options = {fill= colorblindfree11_3}, dashed},
        {draw=colorblindfree11_9, mark = *, mark options = {fill= colorblindfree11_9}, densely dashed},
        {draw=colorblindfree11_8, mark = star, mark options = {fill= colorblindfree11_8}},
        {draw=colorblindfree11_4, mark = star, mark options = {fill= colorblindfree11_4}}
        }
        ]
        \addplot table [x=Eb/N0, y=BEGPSCL8L2DisCRC] {data/n_4096_R_0_5_CRC_16_dist_CRC.txt};
        \addplot table [x=Eb/N0, y=IDABEGPSCL8L2] {data/n_4096_R_0_5_CRC_16_dist_CRC.txt};
        \addplot table [x=Eb/N0, y=IDABEGPSCL8L2both] {data/n_4096_R_0_5_CRC_16_dist_CRC.txt};
        \addplot table [x=Eb/N0, y=BEGPSCL48L2] {data/n_4096_R_0_5_CRC_16_dist_CRC.txt};
      \end{axis}
    \end{tikzpicture}

\caption{FERs of different ways of applying IDA decoding on the BE GPSCL decoder with $S=2$ when decoding the $n=4096$ and $R=0.50$ polar code.}
\label{fig:FER_diff_IDA}
\end{figure}
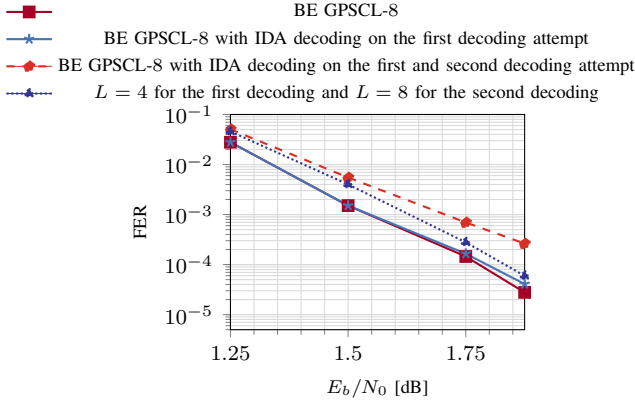

For the second decoding attempt, the \gls{scl} decoder with a list size of $8$ is used since the second decoding attempt is rarely used, and reducing its list size has little impact on the average computational complexity.
Also, retaining the list size in the second decoding attempt can preserve the error correction ability of the \gls{be} decoder.
Fig.~\ref{fig:FER_diff_IDA} shows the \gls{fer} results of applying the \gls{ida} decoding only on the first decoding attempt and applying \gls{ida} decoding on both decoding attempts, and parameters for the \gls{ida} decoding are shown in Table~\ref{tab:IDA_param}.
At a \gls{fer} of $3\times10^{-3}$, applying \gls{ida} decoding to both decoding attempts results in a $0.15\text{ dB}$ degraded decoding performance compared to the \gls{be} \gls{gpscl} decoder.
We can conclude that using a decoder with a list size of $8$ in the second decoding attempt is the key to preserving the error correction performance.
Also, from the results of the \gls{ida} decoding, an accurate bias can be generated without fixing the list size in the first decoding attempt.
For the \gls{scl} decoder, the list size information is only needed starting from the first information bit.
Hence, the \gls{ida} decoding can run in parallel with the \gls{scl} decoding before reaching the first information bit.

As we can assume all-zeros codewords are transmitted~\cite{Li2023improvePAlist}, the average list size of the \gls{ida} decoding can be computed numerically~\cite{Li2023improvePAlist}.
However, when the code length $n$ and $\varphi$ are large, the binomial coefficient and the probability terms in the analytical expression can go beyond the representation range of the data type used by the numerical computation.
For example, in a Python implementation, a large binomial coefficient calculated by the math module can go over (i.e., overflow) the range of the floating-point representation used by the NumPy and scipy packages.
Also, a large exponent on the probability term will produce a result that is smaller than the smallest value that can be represented by the floating-point representation (i.e., underflow) in Python.
\begin{table}
    \centering
    \caption{Parameters for IDA decoding.}
    \begin{tabular}{@{}c c cc c cc c cc c@{}}
    \toprule
    &\multicolumn{2}{c}{$R=0.25$}&\phantom{a}&\multicolumn{2}{c}{$R=0.50$}&\phantom{a}&\multicolumn{2}{c}{$R=0.75$}\\
    \cmidrule{2-3} \cmidrule{5-6} \cmidrule{8-9}
    &$\gamma$ & $\varphi$ && $\gamma$ & $\varphi$  && $\gamma$ & $\varphi$\\
    \midrule
    $n=4096$ &  $0.50$ & $713$ && $0.25$ & $152$ && $0.25$ & $50$\\
    $n=8192$ & $0.25$ & $750$ && $0.25$ & $328$ && $0.25$ & $112$\\ 
    \bottomrule
    \end{tabular}
    \label{tab:IDA_param}
\end{table}

Instead of the direct computation, the probability $\delta$ of using a small list size can be computed in the log domain:
\begin{equation}
\begin{split}
    \delta &= \sum_{i=0}^{\varphi-1} \binom{n}{i}P(|l| > \gamma| x=1)^{n-i}P(|l| \leq \gamma| x=1)^{i}\\
    & = \sum_{i=0}^{\varphi-1} 10^{f(\gamma, i, n, l)}\text{,}
\end{split}
\label{eqn:IDA_exp_L_log10}
\end{equation}
\begin{equation}
\begin{split}
    &f(\gamma, i, n, l) = \log_{10}\left( \binom{n}{i} \right) + i\log_{10}\left(P(|l| \leq \gamma| x=1)\right) \\
    &+ \left( n-i \right) \log_{10}\left(P(|l| > \gamma| x=1)\right)\text{,}
\end{split}
\label{eqn:IDA_exp_L_log10_sub}
\end{equation}
where $x=1$ is the \gls{bpsk} modulated symbol of the code bit $0$.
To consider the computational complexity induced by the second decoding attempt, the average list size of the \gls{ida} \gls{be} \gls{gpscl} is defined as the total list size of the first decoding plus the total list size of the second decoding, and then divided by the number of simulated frames.
\begin{figure*}[ht]
\centering
\begin{tikzpicture}
  \begin{customlegend}[legend columns=3,legend style={at={(0.0,-5.0)},anchor=south west, align=left,draw=none},
  legend entries={
 \scriptsize{\gls{ida} \gls{be} \gls{gpscl}-$8$ $S=2$ $T=2$}, \scriptsize{\eqref{eqn:IDA_exp_L_log10} with parameters in Table~\ref{tab:IDA_param}}
  }
  ]
  \addlegendimage{draw=colorblindfree11_1, mark=square*, fill=colorblindfree11_1}
  \addlegendimage{draw=colorblindfree11_10,mark=*,fill=colorblindfree11_10, densely dashed}
  \addlegendimage{draw=colorblindfree11_2,mark=pentagon*,fill=colorblindfree11_2}
  \addlegendimage{draw=colorblindfree11_11,mark=triangle*,fill=colorblindfree11_11, densely dashed}
  \addlegendimage{draw=colorblindfree11_3,mark=diamond*,fill=colorblindfree11_3}
  \addlegendimage{draw=colorblindfree11_9,mark=star,fill=colorblindfree11_9, densely dashed}

  \addlegendimage{draw=colorblindfree11_4,fill=colorblindfree11_4}
  \addlegendimage{draw=colorblindfree11_8,fill=colorblindfree11_8}

  %\addlegendimage{draw=colorblindfree11_11,fill=colorblindfree11_11}
  \end{customlegend}
\end{tikzpicture}
\vspace{-1.5em}

\captionsetup[subfigure]{font=scriptsize,labelfont=scriptsize}
%\resizebox{.8\textwidth}{!}{
\subfloat[\scriptsize{$n=4096$}]{
    \centering
    \begin{tikzpicture}[spy using outlines={circle,black,magnification=2, connect spies}]
      \begin{axis}[
        footnotesize, width=\columnwidth, height=.42\columnwidth,     
        xmin=0.5, xmax=3.125, xtick={0.50,0.75,...,3.25}, minor x tick num=4,
        ymin=6,  ymax=8.5,
        xlabel=\scriptsize{$E_b/N_0 \text{ [dB]}$}, ylabel=\scriptsize{Avg. $L$},  ylabel near ticks,
        grid=both, grid style={gray!30}, 
        %ymode=log, log basis y={10},
        y tick label style={
                /pgf/number format/fixed,
                /pgf/number format/fixed zerofill,
                /pgf/number format/precision=2
            },
        tick align=outside, tickpos=left, legend columns=1,legend style={
        at={(0.50,1.50)},anchor=center},
        cycle list = {
        {draw=colorblindfree11_1,mark=square*, mark options = {fill=colorblindfree11_1}},
        {draw=colorblindfree11_10,mark=*,mark options = {fill=colorblindfree11_10}, densely dashed}
        %{draw=colorblindfree11_2,mark=pentagon*,mark options ={fill=colorblindfree11_2}},
        %{draw=colorblindfree11_11,mark=triangle*,mark options = {fill=colorblindfree11_11}, densely dashed},
        %{draw=colorblindfree11_3, mark = diamond*, mark options = {fill= colorblindfree11_3}},
        %{draw=colorblindfree11_9, mark = star, mark options = {fill= colorblindfree11_9}, densely dashed},
        %{draw=colorblindfree11_4, mark options = {fill= colorblindfree11_4}}
        }
        ]
        \addplot table [x=0.25Eb/N0, y=0.25AvgL] {data/n_4096_CRC16_IDA_Avg_L.txt};
        \addplot table [x=0.25Eb/N0, y=0.25ExpL] {data/n_4096_CRC16_IDA_Avg_L.txt};
        \node [pin=180:\tiny{$R=0.25$}] at (1.15,6.7) {};
        \addplot table [x=0.50Eb/N0, y=0.50AvgL] {data/n_4096_CRC16_IDA_Avg_L.txt};
        \addplot table [x=0.50Eb/N0, y=0.50ExpL] {data/n_4096_CRC16_IDA_Avg_L.txt};
        \node [pin=180:\tiny{$R=0.50$}] at (1.8,6.7) {};
        \addplot table [x=0.75Eb/N0, y=0.75AvgL] {data/n_4096_CRC16_IDA_Avg_L.txt};
        \addplot table [x=0.75Eb/N0, y=0.75ExpL] {data/n_4096_CRC16_IDA_Avg_L.txt};
        \node [pin=180:\tiny{$R=0.75$}] at (3.0,6.7) {};
      \end{axis}
    \end{tikzpicture}
}
\subfloat[\scriptsize{$n=8192$}]{
    \centering
    \begin{tikzpicture}[spy using outlines={circle,black,magnification=2, connect spies}]
      \begin{axis}[
        footnotesize, width=\columnwidth, height=.42\columnwidth,     
        xmin=0.50, xmax=2.875, xtick={0.50,0.75,...,3.00},
        minor x tick num=4,
        ymin=6,  ymax=8.5,
        xlabel=\scriptsize{$E_b/N_0 \text{ [dB]}$}, 
        %ylabel=\scriptsize{FER},  
        %ylabel near ticks,
        grid=both, grid style={gray!30}, 
        y tick label style={
                /pgf/number format/fixed,
                /pgf/number format/fixed zerofill,
                /pgf/number format/precision=2
            },
        %ymode=log, log basis y={10},
        tick align=outside, tickpos=left, legend columns=1,legend style={
        at={(0.50,1.50)},anchor=center},
        cycle list = {
        {draw=colorblindfree11_1,mark=square*, mark options = {fill=colorblindfree11_1}},
        {draw=colorblindfree11_10,mark=*,mark options = {fill=colorblindfree11_10}, densely dashed}
        %{draw=colorblindfree11_2,mark=pentagon*,mark options ={fill=colorblindfree11_2}},
        %{draw=colorblindfree11_11,mark=triangle*,mark options = {fill=colorblindfree11_11}, densely dashed},
        %{draw=colorblindfree11_3, mark = diamond*, mark options = {fill= colorblindfree11_3}},
        %{draw=colorblindfree11_9, mark = star, mark options = {fill= colorblindfree11_9}, densely dashed},
        %{draw=colorblindfree11_4, mark options = {fill= colorblindfree11_4}}
        }
        ]
        \addplot table [x=0.25Eb/N0, y=0.25AvgL] {data/n_8192_CRC16_IDA_Avg_L.txt};
        \addplot table [x=0.25Eb/N0, y=0.25ExpL] {data/n_8192_CRC16_IDA_Avg_L.txt};
        \node [pin=0:\tiny{$R=0.25$}] at (0.90,6.7) {};
        \addplot table [x=0.50Eb/N0, y=0.50AvgL] {data/n_8192_CRC16_IDA_Avg_L.txt};
        \addplot table [x=0.50Eb/N0, y=0.50ExpL] {data/n_8192_CRC16_IDA_Avg_L.txt};
        \node [pin=0:\tiny{$R=0.50$}] at (1.55,6.7) {};
        \addplot table [x=0.75Eb/N0, y=0.75AvgL] {data/n_8192_CRC16_IDA_Avg_L.txt};
        \addplot table [x=0.75Eb/N0, y=0.75ExpL] {data/n_8192_CRC16_IDA_Avg_L.txt};
        \node [pin=180:\tiny{$R=0.75$}] at (2.8,6.7) {};
      \end{axis}
    \end{tikzpicture}
}
%}
\caption{Average list sizes (Avg. $L$) of the IDA decoding.}
\label{fig:IDA_SCL_Avg_L}
\end{figure*}
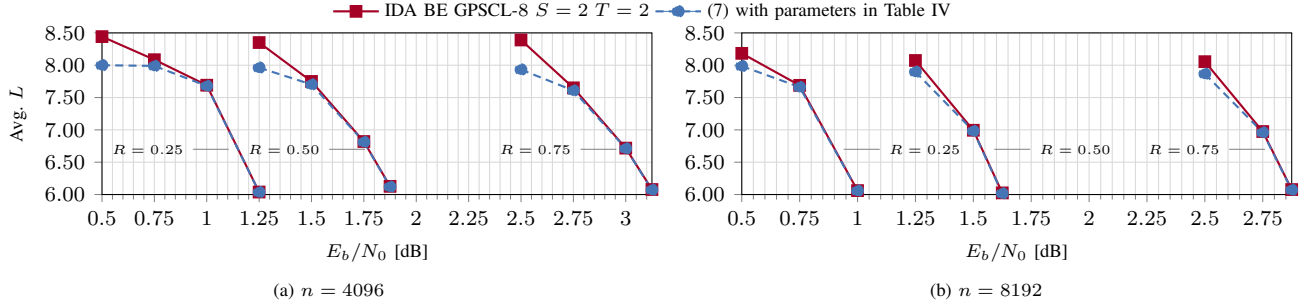

Table~\ref{tab:IDA_param}, Fig.~\ref{fig:partition_BE_SCL}, and Fig.~\ref{fig:IDA_SCL_Avg_L} show parameters used by the \gls{ida} decoding, the \gls{fer} results, and the average list size respectively.
Because a small $\gamma$ value requires a smaller $\varphi$ value than its large counterpart to reach the target expected list size, we choose to use a small $\gamma$ in this work, a set of parameter pairs is generated by~\eqref{eqn:IDA_exp_L_log10} given a target average list size, and parameter pairs in Table~\ref{tab:IDA_param} return good decoding performance according to Fig.~\ref{fig:partition_BE_SCL}.
A small $\varphi$ value implies that a small bit-width is required to represent the $\varphi$ in the quantized model.
We set the parameters $\gamma$ and $\varphi$ such that an average list size of $6$ is achieved at the endpoint of our simulation range.
From Fig.~\ref{fig:partition_BE_SCL} (d), a $0.05\text{ dB}$ degradation is observed when decoding polar codes with $n=8192$ and $R=0.25$ using the \gls{ida} \gls{be} \gls{gpscl} decoder.
For all other code lengths and rates, the degradation is less than $0.05\text{ dB}$.
From Fig.~\ref{fig:IDA_SCL_Avg_L} (a), the average list size returned from the simulations matches the expected list size computed by the new analytical expression~\eqref{eqn:IDA_exp_L_log10} when $E_{b}/N_{0}>1.75 \text{ dB}$ and decoding $n=4096$ and $R=0.50$ polar codes because the number of decoding attempts is close to one in this high $E_{b}/N_{0}$ region.
The same trend exists for codes with different lengths and rates.

In the quantized algorithm, floating-point numbers are rounded to the nearest fixed-point representation.
Hence, the rounding error $\epsilon$ should be taken into account when computing the average list size.
The rounding error $\epsilon$ in this work is half of the smallest value of the fixed-point representation, since the floating-point number is rounded to the nearest fixed-point representation in this work.
The smallest value is $2^{-2} = 0.25$ in this work because two bits are used to represent the fractional part, and $\epsilon = 0.125$.
We denote the \gls{llr} threshold for the quantized algorithm as $\gamma^{'}$.
Then, the parameter $\varphi$ should be chosen under the \gls{llr} threshold $\gamma=\gamma^{'}+\epsilon$ to reach the desired list size.

\begin{table}[t]
    \centering
    \caption{Parameters for quantized IDA decoding.}
    \begin{tabular}{@{}c c cc c cc c cc c@{}}
    \toprule
    &\multicolumn{2}{c}{$R=0.25$}&\phantom{a}&\multicolumn{2}{c}{$R=0.50$}&\phantom{a}&\multicolumn{2}{c}{$R=0.75$}\\
    \cmidrule{2-3} \cmidrule{5-6} \cmidrule{8-9}
    &$\gamma$ & $\varphi$ && $\gamma$ & $\varphi$  && $\gamma$ & $\varphi$\\
    \midrule
    $n=4096$ &  $0.625$ & $888$ && $0.375$ & $229$ && $0.375$ & $75$\\
    $n=8192$ & $0.375$ & $1123$ && $0.375$ & $492$ && $0.375$ & $168$\\ 
    \bottomrule
    \end{tabular}
    \label{tab:IDA_param_quantized}
\end{table}

\begin{table}[t]
    \centering
    \caption{Maximum reductions in computational complexity, which is derived from Fig.~\ref{fig:Avg_op_Quant} and is compared to the SCL-$16$ decoder.}
    \begin{tabular}{@{}l c cc c cc c cc c@{}}
    \toprule
    &\multicolumn{2}{c}{$R=0.25$}&\phantom{a}&\multicolumn{2}{c}{$R=0.50$}&\phantom{a}&\multicolumn{2}{c}{$R=0.75$}\\
    \cmidrule{2-3} \cmidrule{5-6} \cmidrule{8-9}
    &$4096$ & $8192$ && $4096$ & $8192$ && $4096$ & $8192$\\
    \midrule
    \gls{be} &  $3.5\times$ & $3.3\times$ && $3.7\times$ & $3.6\times$ && $4.0\times$ & $3.9\times$\\
    \gls{ida} \gls{be} & $4.7\times$ & $4.4\times$ && $5.1\times$ & $4.9\times$ && $5.4\times$ & $5.2\times$\\ 
    \bottomrule
    \end{tabular}
    \label{tab:max_red}
\end{table}

Table~\ref{tab:IDA_param_quantized}, Fig.~\ref{fig:partition_BE_SCL_quantized}, and Fig.~\ref{fig:IDA_SCL_Avg_L_Quant} show parameters used in the quantized algorithm,  the \gls{fer} of the quantized \gls{ida} \gls{be} \gls{gpscl} decoder, and the average list size respectively.
From Fig.~\ref{fig:partition_BE_SCL_quantized} (b), the quantized \gls{ida} decoding has a $0.05\text{ dB}$ loss in the error correction performance when decoding $n=8192$ and $R=0.25$ polar codes, which is consistent with the loss in the floating-point results shown in Fig~\ref{fig:partition_BE_SCL} (d).
For all other code lengths and rates, the quantized \gls{ida} decoding has less than $0.05\text{ dB}$ loss in the error correction performance compared to the quantized \gls{be} \gls{gpscl} decoder.
From Fig.~\ref{fig:IDA_SCL_Avg_L_Quant}, the computed average list size matches the list sizes returned from the simulations in the high $E_{b}/N_{0}$ region.

\begin{table}[t]
    \centering
    \caption{Latency $\nu$ for IDA decoding and the latency $\Delta$ before reaching the first information bit.}

}
%}
\caption{Average list sizes (Avg. $L$) of the quantized IDA decoding.}
\label{fig:IDA_SCL_Avg_L_Quant}
\end{figure*}

\begin{figure*}[ht]
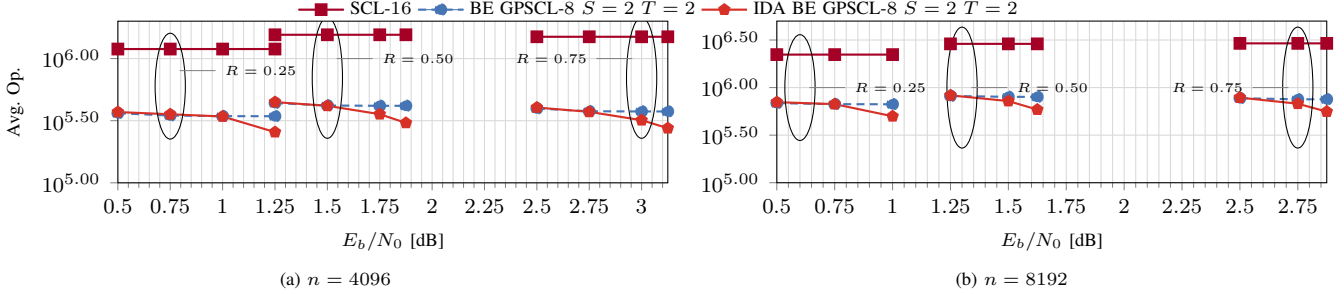

\centering
%
}
%}
\caption{The average number of computational complexity (Avg. Op.) measured by addition, comparison, and selection operations for quantized decoders.}
\label{fig:Avg_op_Quant}
\end{figure*}
\begin{figure*}
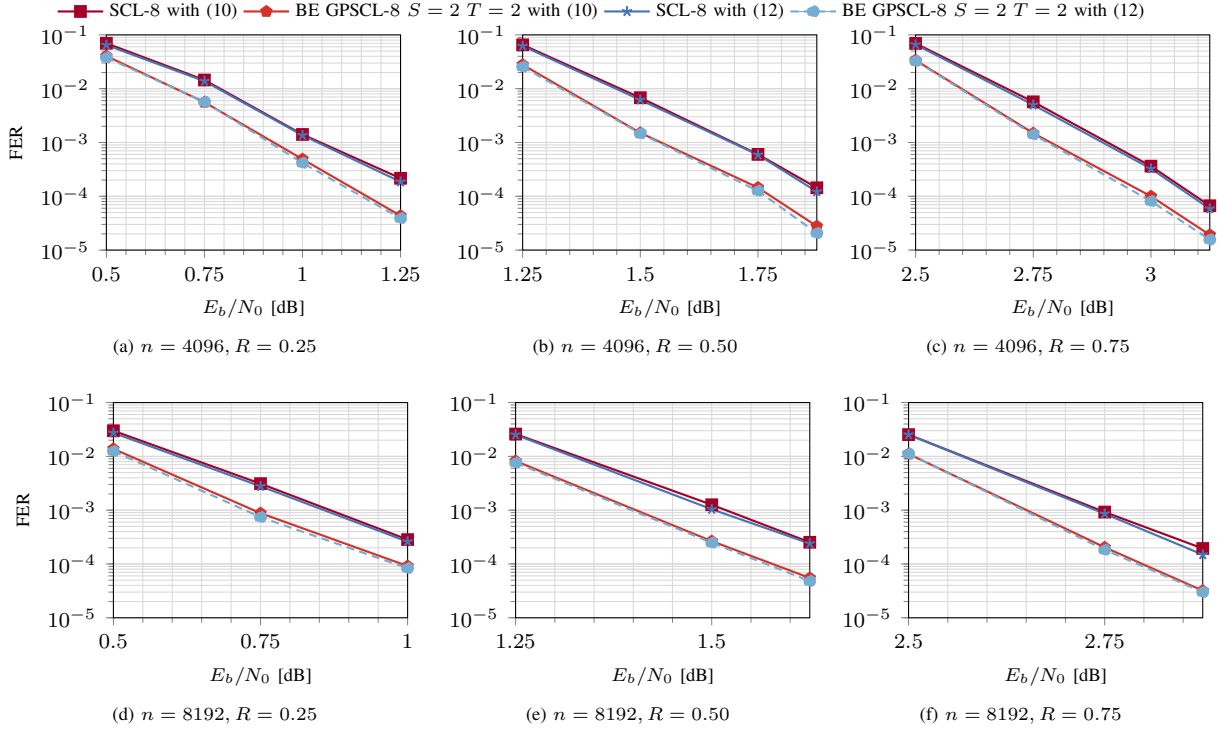

\centering
%
}
\caption{FERs of decoding polar codes using the decoders with the list sphere decoding \gls{pm}~\eqref{eqn:FSSCLSPCListSphere}~\cite{hashemi2017FastSCL} and the exact \gls{pm}~\eqref{eqn:FSSCLSPCcalibrated} respectively.}
\label{fig:cali_BE_SCL}
\end{figure*}

The number of floating-point operations, which are the number of additions and comparisons, is used to calculate the computation complexity of the \gls{scl} decoder in~\cite{doan2022SPSCLrm}.
In this work, we also count the number of fixed-point operations, which are the number of additions, the number of comparisons, and the number of selections, used by the decoders to calculate the computational complexity.
In the fast \gls{scl} decoder, the selection operation for a \gls{pm} can be viewed as comparing the number of relatively higher \gls{pm}s to the list size $L$.
If the number of relatively higher \gls{pm}s is larger than $L$, then this \gls{pm} is one of the top $L$ \gls{pm} out of the $2L$ possible paths, and it will be kept in the list.
Hence, the selection operation can be viewed as a comparison.

The min-sum approximation used by the left tree traversal needs one comparison, and the right tree traversal uses one addition.
In special nodes, an addition is needed to update the \gls{pm} with a zero or a penalty.
For \gls{spc} nodes, one addition is required to initialize the \gls{pm}, and two additions are required to update the \gls{pm}.
In the fast \gls{scl} decoder, the sorter compares all \gls{pm} with each other simultaneously~\cite{hashemi2017FastSCL}, let the number of \gls{pm} be $L_{\text{sorter}}$, and each \gls{pm} uses $L_{\text{sorter}}-1$ comparisons and one extra selection to check whether it is the top $L$ candidates.
Hence, the sorter requires $L_{\text{sorter}}^{2}$ comparisons and selection.
Biased-enhancement takes $n$ additions.
The number of operations required by the \gls{ida} decoding composes of: $n$ comparisons on the received \gls{llr} vector, $\sum_{i=0}^{\log_{2}\left( n\right)-1}2^{i}=n-1$ additions for counting the number of unreliable \gls{llr}s, and one comparison with $\varphi$.
Fig.~\ref{fig:Avg_op_Quant} shows the average number of operations, and the maximum reduction is shown in Table~\ref{tab:max_red}.
The \gls{be} \gls{gpscl} decoder with a list size of $8$ has at most $4\times$ fewer operations than the \gls{scl} decoder with a list size of 16.
The reason is that the number of sorting operations increases quadratically with the list size.
A reduction in the computational complexity of up to $5.4\times$ is returned by the \gls{ida} \gls{be} \gls{gpscl} with a list size of $8$ compared to the \gls{scl} decoder with a list size of $16$.

The latency of the \gls{scl} decoder is modeled as follows.
The latency of decoding special nodes follows the metric in~\cite{SRSCL}.
Left and right tree traversal, and re-encoding take one clock cycle.
The \gls{ida} decoding uses one clock cycle for \gls{llr} comparisons, $\log_{2}(n)$ cycles in the adder tree for counting, and one clock cycle to compare with $\varphi$.
The latency is shown in Table~\ref{tab:IDA_overhead}.
When running the \gls{ida} decoding in parallel with the \gls{scl} decoder, we can see that no overhead is produced except for polar codes with $n=8192$ and $R=0.75$, and this exception is highlighted in the bold font in Table~\ref{tab:IDA_overhead}.

\section{The Calibrated Path Metric on SPC Nodes and the Theoretical Analysis on the Bias}
\label{sec:ana_pm_bias}
In this section, we investigate the source of error due to the approximated \gls{pm}~\cite{hashemi2016SSCL} for the \gls{spc} node in the fast \gls{scl} decoder~\cite{hashemi2017FastSCL}.
Approximation errors, such as the min-sum approximation~\cite{chen2005reduced,chisnevski2025minsumpolar}, are analyzed when designing decoders for polar codes and \gls{ldpc} codes.
We find that using a calibrated \gls{pm}, which is equivalent to the Chasing decoding metric, in \gls{spc} nodes returns better error correction performance while using the same complexity as using the approximated \gls{pm}~\cite{hashemi2017FastSCL}.
Based on this calibrated \gls{pm} on the \gls{spc} nodes, we prove that the bias in the \gls{be} decoder moves the received \gls{llr} toward valid polar codewords with a high likelihood, which explains the improved \gls{fer} brought by the \gls{be} \gls{scl} decoder.
\subsection{Analysis of Path Metrics for \gls{spc} Nodes}

A list sphere decoding approach~\cite{sphere_polar}, which is a hardware-friendly variant of the maximum-likelihood sphere decoder for polar codes~\cite{hashemi2015listSD}, is used to decode the \gls{spc} nodes~\cite{hashemi2016SSCL}.
The lossless decoding approach is to evaluate the parity check of two bits simultaneously and generate bit flips that satisfy the parity check~\cite{hashemi2016SSCL}.
However, the lossless approach requires $\binom{N_{v}}{2}$ time steps for a \gls{spc} node with a length of $N_{v}$, which will cause a large decoding latency when $N_{v}\geq 8$~\cite{hashemi2016SSCL}.
An approximated decoding approach, which only considers two-bit combinations of the least reliable bit and all other bits, is proposed in~\cite{hashemi2016SSCL}. %, and it reduces the time step to at most $N_{v}-1$.
Later, it is proved in~\cite{hashemi2017FastSCL} that only the first $\min\left(L, N_{v}\right)$ least reliable code bits should be enumerated for the \gls{spc} nodes, which reduces the time steps to $\min\left(L, N_{v}\right)$.
However, the error induced by this approximated list sphere decoding has not been investigated in the literature.

For the approximated \gls{pm} update, given the \gls{pm} ($\text{PM}_{-1}$) from the previous node, the \gls{spc} node's \gls{pm} is initialized by
\begin{equation}
    \text{PM}_{0} =
    \begin{cases}
        \text{PM}_{-1} + |\alpha_{i_{\min}}|\text{, if }\Gamma=1\text{,}\\
        \text{PM}_{-1}\text{, otherwise,}
    \end{cases}
    \label{eqn:SPC_pm_ini}
\end{equation}
where $\alpha_{i_{\min}}\in \bm{\alpha}$ is the \gls{llr} with the smallest magnitude in the \gls{llr} vector $\bm{\alpha}$, $i_{\min}$ is the index of this least reliable \gls{llr}, and $\Gamma$ is the parity check.
For every enumerated bit $i$, given the parity check $\Gamma$ of the hard decision of the received \gls{llr} vector $\bm{\alpha}$ and the new hard estimation $\beta_{i}$ of $\alpha_{i}$, the \gls{pm} is updated by 
\begin{equation}
    \text{PM}_{i} = 
    \begin{cases}
        \text{PM}_{i-1} + |\alpha_{i}| + \left( -1 \right)^{\Gamma} |\alpha_{i_{\min}}|\text{, }\left( -1\right)^{\beta_{i}} \neq \text{sign}\left( \alpha_{i} \right)\text{,}\\
        \text{PM}_{i-1}\text{, otherwise.}
    \end{cases}
    \label{eqn:FSSCLSPCListSphere}
\end{equation}
If we reformulate the update of the \gls{pm}, we can see that
\begin{equation}
    \text{PM}_{i}=\text{PM}_{0} + \sum_{j \in \mathbbm{j}} |\alpha_{j}| + \left( -1\right)^{\Gamma}\left|  \mathbbm{j} \right| \left| \alpha_{i_{\min}} \right|\text{,}
\end{equation}
where $\mathbbm{j}$ is the set of indices of flipped bits.
This update is only exact when the list size $L=2$~\cite{hashemi2016SSCL}.
When the list size $L>2$, candidates from paths with $\Gamma=1$ are more likely to be selected when using this approximated \gls{pm} update because they tend to have a smaller \gls{pm}.

The \gls{pm} update for the approximated decoding can be calibrated to the exact \gls{pm}, which is equivalent to Chase decoding for \gls{spc} nodes~\cite{sarkis2016FastSCL}, by
\begin{equation}
\begin{split}
    &\text{PM}_{i} =\\ 
    &\begin{cases}
        \text{PM}_{i-1} + |\alpha_{i}| -\left( -1 \right)^{\Gamma} (-1)^{\text{wt}\left(\text{diff}\right)} |\alpha_{i_{\min}}|\text{, if } \text{flip}=\text{true}\text{,}\\
        \text{PM}_{i-1}\text{, otherwise,}
    \end{cases}
\end{split}
\label{eqn:FSSCLSPCcalibrated}
\end{equation}
$\text{flip}:=\left( -1\right)^{\beta_{i}} \neq \text{sign}\left( \alpha_{i} \right)$, and
\begin{align*}
    \text{diff}:=\left\{\left|\beta_{i} - \text{HD}\left( \alpha_{i}\right)\right|\right | i\in \{1, 2, ..., N_{v} \}\setminus i_{\text{min}}\}
\end{align*}
is the bit-flip pattern applied, and $\text{wt}\left( \text{diff}\right)$ is the number of flipped bits.
The following proposition shows the equivalence.
\begin{prop}
    The \gls{pm} update~\eqref{eqn:FSSCLSPCcalibrated} is equivalent to the \gls{pm} update according to the Chase decoding in~\cite{sarkis2016FastSCL}.
    \label{prop:calibrated_PM_Chase_PM}
\end{prop}
\begin{proof}
    From~\eqref{eqn:SPC_pm_ini}, it can be checked that the \gls{pm} update for the \gls{spc} node is firstly initialized according to the parity-check result and the least reliable \gls{llr}.
    
    Then, according to~\eqref{eqn:FSSCLSPCcalibrated}, the \gls{pm} for the \gls{spc} node is adaptively adjusted by a value of $|\alpha_{\text{min}}|$ according to the enumeration patterns applied on \gls{llr}s except the \gls{llr} with the smallest magnitude such that the parity-check constraint is always satisfied.
    
    Also, in~\eqref{eqn:FSSCLSPCcalibrated}, the \gls{pm} is updated according to whether the bit estimation has the same sign as the \gls{llr} or not.
    Hence, we can conclude that the \gls{pm} update~\eqref{eqn:FSSCLSPCcalibrated} is equivalent to the \gls{pm} update according to the Chase decoding metric in~\cite{sarkis2016FastSCL}.
\end{proof}
Open-source implementations like~\cite{yongrunpolar} and the other work~\cite{zhao2023calipm} have adopted this calibrated update for the \gls{spc} nodes.
When the parity check $\Gamma=1$, the penalty caused by flipping the least reliable bit is deducted if an odd number of bit flips is applied on bits other than the least reliable bit.
The penalty caused by flipping the least reliable bit is included if an even number of bit flips is applied on bits other than the least reliable bit.
When the parity check $\Gamma=0$, the opposite operation is performed.

Comparisons, which are missing from the literature, between the approximated and the exact \gls{pm}s are shown in Fig.~\ref{fig:cali_BE_SCL}.
From Fig.~\ref{fig:cali_BE_SCL} (b), the \gls{scl} and \gls{be} \gls{gpscl} decoder with the exact \gls{pm}~\eqref{eqn:FSSCLSPCcalibrated} returns better error correction performance than their counterpart using the approximated \gls{pm}~\eqref{eqn:FSSCLSPCListSphere} when decoding length $4096$ and rate $0.5$ polar codes.
However, the improvement in the error correction is less than $0.05\text{ dB}$ using our simulation setting.
The same trend exists for all other code lengths and rates.
We provide a proof of the maximum number of bit estimations under this exact \gls{pm}~\eqref{eqn:FSSCLSPCcalibrated} in Theorem~\ref{thm:spcPMcali}.
By Theorem~\ref{thm:spcPMcali}, the exact \gls{pm} and the approximated \gls{pm} have the same time step for \gls{spc} nodes~\cite{hashemi2017FastSCL}.
\begin{theorem}
    In the \gls{scl} decoding, the maximum number of bit estimations in a \gls{spc} node with a length $N_{v}$ and the \gls{pm}~\eqref{eqn:FSSCLSPCcalibrated} is bounded by $\min\left(L, N_{v}\right)$.
    \label{thm:spcPMcali}
\end{theorem}
\begin{proof}
    Assuming we have the sorted soft information $|\alpha_{1}|\leq |\alpha_{2}| \leq ... \leq |\alpha_{N_{v}}|$, and we have the following two cases:
    
    \textbf{I}), When $\Gamma = 1$, an odd number of bit flips is performed.
    Considering one bit is flipped, there are $L-1$ \gls{pm}s
    \begin{align*}
        \left\{ \text{PM}_{1}=|\alpha_{1}|, \text{PM}_{2}=|\alpha_{2}|, ..., \text{PM}_{L-1} = |\alpha_{L-1}| \right\}
    \end{align*}
    that are smaller than the \gls{pm} of flipping the $L$-th bit ($\text{PM}_{L}=|\alpha_{L}|$).
    Hence, bit flipping among the $L$ least reliable bits is sufficient to generate $L$ small \gls{pm}s.

    \textbf{II}), When ($\Gamma = 0$), either no bit flip or an even number of bit flips is performed.
    Considering flipping the first (least reliable bit) and all other bits, there are $L-1$ \gls{pm}s
    \begin{align*}
        \left\{\text{PM}_{1}=0, \text{PM}_{2}=|\alpha_{1}|+|\alpha_{2}|, ...,\text{PM}_{L-1} = |\alpha_{1}|+|\alpha_{L-1}|\right\}
    \end{align*}
    that are smaller than the \gls{pm} of flipping the first and the $L$-th bit ($\text{PM}_{L}=|\alpha_{1}|+|\alpha_{L}|$).
    Hence, bit flipping among the $L$ least reliable bits is sufficient to generate $L$ small \gls{pm}s.
     Hence, the maximum number of bit estimations is $\min\left(L, N_{v}\right)$.
\end{proof}
\subsection{Analysis of the Bias from the \gls{be} \gls{scl} Decoding}
It is shown in~\cite{yang2025biassuccessive} that the bias, which applies to the received \gls{llr} that has a small magnitude and a different sign than the decoded codeword from the \gls{sc} decoding, moves the received \gls{llr} toward the \gls{ml} solution (with a small deviation from the received 
\gls{llr} vector).
In this work, we prove that the bias applied to the \gls{be} \gls{scl} decoder moves the received \gls{llr} vector toward valid polar codewords with a small \gls{pm}, and the \gls{fer} gain of the \gls{be} \gls{scl} decoder is explained by the improved observable (\gls{llr}s) due to the bias.
\begin{theorem}
    The bias applied to the \gls{be} \gls{scl} decoder moves the received \gls{llr} vector toward valid polar codewords that have a high likelihood.
    \label{thm:be_trajectory}
\end{theorem}
\begin{proof}
    It is shown in~\cite{hashemi2016SSCL} that
    \begin{equation}
        \begin{split}
            & \frac{1}{2} \left( \text{sign}\left( \alpha_{0}^{l}\right) \alpha_{0}^{l} - \eta_{0}^{l} \alpha_{0}^{l} + \text{sign}\left( \alpha_{0}^{r}\right) \alpha_{0}^{r} - \eta_{0}^{r} \alpha_{0}^{r} \right)\\
            &= \frac{1}{2} \left( |\text{sgn}\left( \alpha_{0}^{l}\right) \alpha_{0}^{l} - \eta_{0}^{l} \alpha_{0}^{l}| + |\text{sgn}\left( \alpha_{0}^{r}\right) \alpha_{0}^{r} - \eta_{0}^{r} \alpha_{0}^{r}| \right)\\
            & = \frac{1}{2} \left( \text{sign}\left( \alpha_{1}\right) \alpha_{1} - \eta_{1} \alpha_{1} + \text{sign}\left( \alpha_{2}\right) \alpha_{2} - \eta_{2} \alpha_{2} \right)\\
            &= \frac{1}{2} \left( |\text{sign}\left( \alpha_{1}\right) \alpha_{1} - \eta_{1} \alpha_{1}| + |\text{sign}\left( \alpha_{2}\right) \alpha_{2} - \eta_{2} \alpha_{2}| \right)\text{,}
            \end{split}
            \label{eqn:recusive_update_rule_list}
        \end{equation}
        where $\alpha_{0}^{l}$ and  $\alpha_{0}^{r}$ are the \gls{llr}s for the left child and the right child respectively, $\eta_{0}^{l}$ and $\eta_{0}^{r}$ are the symbol ($\left\{ -1,1\right\}$) estimations of the left and right child, and  $\eta_{1}$ and $\eta_{2}$ are the symbol estimations for upper-stage \gls{llr}s $\alpha_{1}$ and $\alpha_{2}$ respectively.
        Also, by~\cite[Thm. 2, Thm. 4, Thm. 6]{hashemi2016SSCL}, the \gls{pm} update for the rate-0, rate-1, and the repetition nodes follow
        \begin{equation}
            \text{PM}_{i}=\text{PM}_{i-1}+\frac{1}{2}\sum_{i=1}^{N_{v}}\text{sign}\left(\alpha_{i}\right)\alpha_{i} - \eta_{i}\left(\alpha_{i}\right)\text{.} 
            \label{eqn:PM_update_general}
        \end{equation}
        We can see that the calibrated \gls{pm} update~\eqref{eqn:FSSCLSPCcalibrated} for the \gls{spc} node is equivalent to the Chase decoding metric from Proposition~\ref{prop:calibrated_PM_Chase_PM}, so the \gls{pm} update on the \gls{spc} node is equivalent to~\eqref{eqn:PM_update_general} according to~\cite{sarkis2016FastSCL}.

        Given that the \gls{pm} for the rate-0, rate-1, the repetition, and the \gls{spc} nodes is updated according to~\eqref{eqn:PM_update_general}~\cite{hashemi2016SSCL} and the recursive update~\eqref{eqn:recusive_update_rule_list}, the \gls{pm} of candidate codewords in the list returned from the \gls{scl} decoding is
        \begin{equation}
            \text{PM}=\frac{1}{2}\sum_{i=1}^{n}\text{sign}\left(l_{i}\right)l_{i} -1^{\hat{c}_{i_{p}}}\left(l_{i}\right)=\sum_{ i\in \mathcal{Z} } \left| l_{i}\right|\text{,} 
            \label{eqn:PM_update_general_equivalence}
        \end{equation}
        where $\mathcal{Z}:=\{i|-1^{\hat{c}_{i_{p}}}\neq \text{sign}(l_{i})\wedge i\in \left\{ 1, 2,...,n \right\} \}$,
        $\hat{c}_{i_{p}}$ is the $i$-th code bit estimation in the candidate codeword $p$, and $l_{i}$ is the $i$-th received \gls{llr}.
        According to~\cite{qu2024constituent}, the \gls{pm}~\eqref{eqn:PM_update_general_equivalence} is equivalent to the likelihood of the decoded codeword, and the \gls{ml} codeword will have the least \gls{pm}/the highest likelihood.

        According to~\eqref{eqn:PM_update_general_equivalence}, we can conclude the bias update~\eqref{eqn:perturb_value_bias} applied on the all-agreed and the all-disagreed code bits moves the received \gls{llr} vector toward valid polar codewords with a small \gls{pm} or, equivalently, a high likelihood.
\end{proof}

\section{Conclusion}
\label{sec:conclusion}
By ablation studies, we propose the \gls{be} \gls{scl} decoder that removes the random noise.
The \gls{gpscl} decoder is used to reduce the memory usage of the \gls{be} decoder.
Compared to the \gls{scl} decoder with a list size of $16$, the $67\%$ reduction in memory usage is achieved by the \gls{be} \gls{gpscl} decoder with a list size of $8$.
The \gls{ida} decoding is applied to reduce the computational complexity of the \gls{be} \gls{gpscl} decoder.
Compared to the \gls{scl} decoder with a list size of $16$, the \gls{ida} \gls{be} \gls{gpscl} decoder with a list size of $8$ has a reduction in computational complexity of up to $5.4\times$ while having at most $0.05\text{ dB}$ degraded decoding performance.
We theoretically prove that the bias in \gls{be} \gls{scl} moves the received \gls{llr} vector toward valid polar codewords with a high likelihood, and explain the \gls{fer} gain in the \gls{be} \gls{scl} decoder.

\FloatBarrier

\bibliographystyle{IEEEtran}
\bibliography{IEEEabrv,reference.bib}{}

\end{document}